\documentclass[11pt]{amsart}

\input txdtools

\usepackage{amsmath,mathtools}
\usepackage{amssymb}
\usepackage{mathrsfs}
\usepackage{amsfonts}
\usepackage{graphicx}
\usepackage{texdraw}
\usepackage{graphpap}
\usepackage{enumitem}
\usepackage{chngcntr}

\usepackage{wasysym}
\usepackage{color}
\usepackage[OT2,T1]{fontenc}
\usepackage{tikz}
\usepackage{pgfplots}
\usetikzlibrary{decorations.markings}
\pgfplotsset{width=7cm,compat=1.8}
\usetikzlibrary{spy}
\usepackage{verbatim}

\newcommand{\boundellipse}[3]
{(#1) ellipse (#2 and #3)
}

\def\pa{\partial}
\def\bp{\bar\partial}

\def\bfK{\mathbf{K}}
\def\bfk{\mathbf{k}}
\def\bfR{\mathbf{R}}

\def\C{\mathbb{C}}
\def\D{\mathbb{D}}
\def\R{\mathbb{R}}
\def\T{\mathbb{T}}
\def\LL{\mathbb{L}}
\def\rmn{\mathrm{n}}
\def\rmt{\mathrm{t}}

\def\calW{\mathcal{W}}

\def\calQ{\mathscr{Q}}
\def\calD{\mathcal{D}}
\def\calH{\mathscr{H}}

\def\erfc{\operatorname{erfc}}
\def\Int{\operatorname{Int}}

\def\z{\zeta}
\def\w{\eta}
\def\Lap{\Delta}
\def\normal{\mathrm{n}}

\newcommand{\Prob}{{\mathbb P}}
\newcommand{\fii}{{\varphi}}
\newcommand{\Expe}{{\mathbb E}}
\newcommand{\1}{{\mathbf{1}}}
\newcommand{\re}{\operatorname{Re}}
\newcommand{\im}{\operatorname{Im}}

\newcommand{\dist}{\operatorname{dist}}

\newcommand{\param}{\mathbf{c}}

\newcommand{\sfun}{b}

\newcommand{\fif}{\fii}

\newcommand{\loc}{\operatorname{loc}}

\newcommand{\eqpot}{\check{Q}}

\renewcommand{\T}{{\mathbb T}}
\renewcommand{\L}{{\mathbb L}}

\newcommand{\Pc}{\operatorname{Pc}}

\renewcommand{\d}{{\partial}}
\newcommand{\dbar}{\bar{\partial}}

\theoremstyle{plain}
\newtheorem{thm}{Theorem}[section]
\newtheorem{lem}[thm]{Lemma}

\theoremstyle{definition}
\newtheorem*{def*}{Definition}
\newtheorem*{eg*}{Example}

\theoremstyle{remark}
\newtheorem*{rmk*}{Remark}

\numberwithin{equation}{section}

\begin{document}

\title[A scale of boundary confinements]{On boundary confinements for the Coulomb gas}

\author{Yacin Ameur}

\address{Yacin Ameur\\
Department of Mathematics\\
Faculty of Science\\
Lund University\\
P.O. BOX 118\\
221 00 Lund\\
Sweden}
\email{Yacin.Ameur@maths.lth.se}

\author{Nam-Gyu Kang}
\address{Nam-Gyu Kang\\
School of Mathematics\\
Korea Institute for Advanced Study\\
85 Hoegiro\\
Dongdaemun-gu\\
Seoul 02455\\
Republic of Korea}
\email{namgyu@kias.re.kr}

\author{Seong-Mi Seo}
\address{Seong-Mi Seo\\
School of Mathematics\\
Korea Institute for Advanced Study\\
85 Hoegiro\\
Dongdaemun-gu\\
Seoul 02455\\
Republic of Korea}
\email{seongmi@kias.re.kr}



\thanks{Nam-Gyu Kang was partially supported by Samsung Science and Technology Foundation (SSTF-BA1401-51) and by a KIAS Individual Grant(MG058103) at Korea Institute for Advanced Study.
Seong-Mi Seo was partially supported by a KIAS Individual Grant (MG063103) at Korea Institute for Advanced Study and by a National Research Foundation of Korea Grant funded by the Korea government (No. 2019R1F1A1058006). 
}

\begin{abstract} We introduce a family of boundary confinements for Coulomb gas ensembles, and study them in the two-dimensional determinantal case of random normal matrices.
The family interpolates between the free boundary and hard edge cases, which
have been well studied in various random matrix theories. The confinement can also be relaxed beyond the free boundary to produce ensembles with fuzzier boundaries, i.e.,
where the particles are more and more likely to be found outside of the boundary. The resulting ensembles are investigated with respect to
scaling limits and distribution of the maximum modulus. In particular, we prove existence of a new point field - a limit of scaling limits to the ultraweak
point when the droplet ceases
to be well defined.
\end{abstract}

\maketitle

\section{Introduction and main results}

In the theory of Coulomb gas ensembles, it is natural to consider different kinds of boundary confinements. The most well-known examples are the ``free boundary'', where particles are admitted to range freely outside of the droplet,
and the ``hard edge'', where they are completely confined to it. On the other hand, notions of weakly confining potentials have attracted attention recently, where
the boundary is softer than a free boundary, i.e., particles are more likely to be found outside of the boundary. In this note,
we introduce a one-parameter family of edge confinements, ranging all the way between an idealized ``ultraweak'' edge and a hard edge.

\begin{figure}[ht]\label{fig1}
\begin{center}
\includegraphics[width=.25\textwidth]{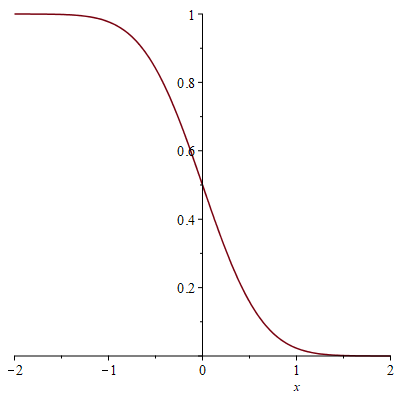}
\hspace{.02\textwidth}
\includegraphics[width=.25\textwidth]{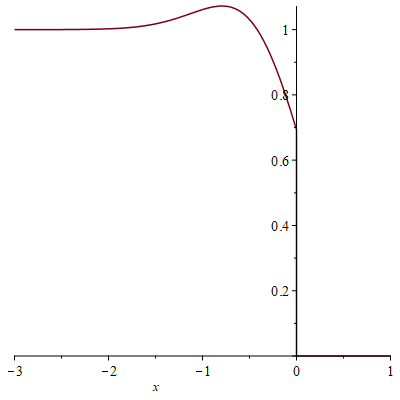}
\hspace{.02\textwidth}
\includegraphics[width=.25\textwidth]{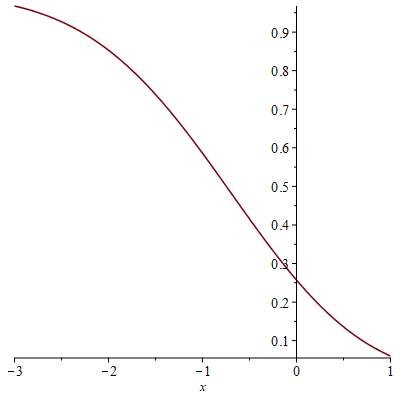}
\end{center}
\caption{Density profiles $R(x)$ at the free boundary, the hard edge, and the ultraweak edge, respectively.}
\end{figure}

Our construction can be applied to general
Coulomb gas ensembles in any dimension and for any inverse temperature $\beta$.
However, we shall here be content to develop the theory only
in the determinantal, two-dimensional case,
i.e., we will consider ensembles of
eigenvalues of random normal matrices.

The study of universality in free boundary ensembles has been the focus of several recent works \cite{AKM,AKMW,HW17}. Notably, in the paper \cite{HW17}, it is shown that with free boundary confinement, the point field with intensity function $R(z)=\fif(z+\bar{z})$ appears universally (i.e., for a ``sufficiently large'' class of ensembles) when rescaling about a regular boundary point, where $\fif$, the ``free boundary function'', is given by
\begin{equation}\label{ffun}\fif(z)=\sfun_1(z):=\frac 1 {\sqrt{2\pi}}\int_{-\infty}^{0}e^{-(z-t)^2/2}\, dt=\frac 1 2 \erfc\frac z{\sqrt{2}}.\end{equation}

For the hard edge Ginibre ensemble,
a direct computation with the orthogonal polynomials in \cite[Section 2.3]{AKM} shows that,
under a natural scaling about a boundary point, the point process
of eigenvalues converges to the determinantal point field determined by the 1-point function $R(z)=\sfun_\infty(z+\bar{z})\cdot \1_\L(z)$
where $\1_\L$ is the indicator function of the left half plane $\L=\{\re z<0\}$ and where $\sfun_\infty$, the ``hard edge plasma function'', is defined by
\begin{equation}\label{hfun} \sfun_\infty(z)=\frac 1 {\sqrt{2\pi}}\int_{-\infty}^0\frac {e^{-(z-t)^2/2}}{\fif(t)}\, dt.\end{equation}
As far as we know, this function appeared first in the physical paper \cite{Sm} from 1982, cf. \cite[Section 15.3.1]{Fo};
see Figure 1.

We shall introduce a scale of point-processes depending on a \textit{confinement-parameter} $\param$ such that the values $\param=1$ and $\param=\infty$ correspond to the free boundary and the hard edge, respectively.

We note at this point that the functions $\sfun_1$ and $\sfun_\infty$ both have the convolution structure
\begin{equation}\label{def:soft}
\sfun_{\param}(z) = \frac{1}{\sqrt{2\pi}}\int_{-\infty}^{0}
\frac{e^{-(z-t)^2/2}}{\Phi_{\param}(t)} \, dt,
\end{equation}
where $\Phi_1=1$ and $\Phi_\infty=\fii$. We shall show that if $\param>0$, then the choice
\begin{equation}\label{defphi}\Phi_\param(t):=\fii(t)+\frac 1 {\sqrt{\param}}(1-\fii(\frac t{\sqrt{\param}}))\exp\left\{\frac {(1-\param)t^2}{2\param}\right\}\end{equation}
leads to a scale of new determinantal point fields with $1$-point intensities
\begin{equation}\label{1pt}R(z)=R^{(\param)}(z)=\sfun_\param(z+\bar{z})\exp\left\{2(1-\param)(\re z)_+^2\right\},\end{equation}
where we write $x_+=\max\{x,0\}$  and (to avoid bulky notation) $x_+^2$ instead of $(x_+)^2$.

We shall find that these point fields emerge naturally as scaling limits about regular boundary points of the droplet, if we set up appropriate boundary confinements.
In the limit as $\param\to 0+$ we will establish existence of a new point field in the critical case where the droplet ceases to be well defined. This
point field might be said
to model an idealized ultraweak edge; its density profile is depicted in Figure \ref{fig1} and also in Figure 3.

\begin{rmk*}
Hard edge ensembles are well known in the Hermitian theory, where they are usually associated with the Bessel kernel \cite{Fo,Fo2,TW2}. Another possibility, a ``soft/hard edge'', appears when a soft edge is replaced by a hard edge cut. This situation was studied by Claeys and Kuijlaars in the paper \cite{CKu08}. In this case a Painlev\'{e} II kernel arises instead of a Bessel kernel. The hard edges in the present note (in the case $\param=\infty$) are actually of the soft/hard type, but
to keep our terminology simple, we prefer to use the adjective ``hard''. (We follow the papers \cite{AKM,AKMW} in this connection.)
\end{rmk*}

\subsection{Basic setup}
Fix a function (``external potential'') $Q:\C\to\R\cup\{+\infty\}$ and write $\Sigma=\{Q<+\infty\}$. We assume that $\Int \Sigma$ be dense in $\Sigma$ and that $Q$ be lower semicontinuous on $\C$ and real-analytic on $\Int\Sigma$ and ``large'' near $\infty$:
\begin{equation}\label{growth}\liminf_{\zeta\to\infty}\frac {Q(\zeta)}{\log|\zeta|^2}>1.\end{equation}

We next form the \textit{equilibrium measure} $\sigma$ in external potential $Q$, namely the measure $\mu$ that minimizes the weighted energy
\begin{equation}\label{q-en} \displaystyle I_Q[\mu]=\iint_{\C^2}\log\frac 1 {|\zeta-\eta|}\, d\mu(\zeta)d\mu(\eta)+\int_\C Q\, d\mu\end{equation}
amongst all compactly supported Borel probability measures $\mu$ on $\C$.
It is well known \cite{ST} that $\sigma$ is unique, is absolutely continuous, and takes the form
\begin{equation}\label{EQMEAS}d\sigma(\zeta)=\Lap Q(\zeta)\cdot\1_S(\zeta)\, dA(\zeta),\end{equation}
where $S$ is a compact set which we call the \textit{droplet} in potential $Q$. Here and henceforth we use the convention that $\Lap=\d\dbar$ denotes $1/4$ times the usual Laplacian, while $dA=dxdy/\pi$ is Lebesgue measure normalized so that the unit disk has measure $1$.

We will in the following assume that $S\subset\Int\Sigma$ and that $S$ be connected.
Under our assumptions, $S$ is finitely connected and the boundary $\d S$ is a union of a finite number of real-analytic arcs, possibly with finitely many singular points which can be certain types of cusps and/or double points. See e.g., \cite{AKMW,LM}.

We shall consider the \textit{outer boundary}
$\Gamma=\d \Pc S,$
where the \textit{polynomially convex hull} $\Pc S$ is the
union of $S$ and the bounded components of $\C\setminus S$.
Thus $\Gamma$ is a Jordan curve having possibly finitely many singular points. We shall assume that $\Gamma$ be \textit{everywhere regular}, i.e., that there are no singular points on $\Gamma$.

\begin{rmk*} We can do with weaker assumptions: for instance concerning regularity it suffices to assume that $Q$ be real-analytic in a neighborhood of $\Gamma$ and, say, $C^2$-smooth on $\Int\Sigma$.
\end{rmk*}

We next recall a basic potential-theoretic construct: the \textit{obstacle function} $\check{Q}(\zeta)$ in external potential $Q$. This function can be defined in several ways, e.g., as the maximal subharmonic function which is less than $Q$ on $\C$ and grows at most like $\log|\zeta|^2 + O(1)$ at infinity,
or as $\gamma-2U^\sigma(\zeta)$ where
$U^\sigma(\zeta)=\int_\C \log\frac1 {|\zeta-\eta|}\, d\sigma(\eta)$ is the logarithmic potential of $\sigma$, and $\gamma$ is a (``modified Robin's'') constant.

The properties of $\check{Q}$ to be used below are (i) $\check{Q}=Q$ on $S$ and $\check{Q}\le Q$ everywhere, (ii) $\check{Q}$ is $C^{1,1}$-smooth on $\C$ and harmonic
on the complement $S^c$ of $S$, (iii) $\check{Q}(\zeta)= \log|\zeta|^2+O(1)$ as $\zeta\to\infty$. (See \cite{HM,ST} for proofs.)

In general, it might happen that $\check{Q}=Q$ on some points of the complement $S^c$, called ``shallow points'' in \cite{HM}. With a mild restriction, we will assume that no such points exist, i.e., we assume that the droplet $S$ equals to the coincidence set
$\{Q=\check{Q}\}$ and that $\check{Q}<Q$ everywhere on $S^c$.

After these proviso, we introduce our main object of study.

We fix a number (the ``confinement-constant'') $\param$ with $0<\param<\infty$
and consider the modified ($C^{1,1}$-smooth) potential $Q^{(\param)}$ defined by
$$Q^{(\param)}(\zeta)={\param}Q(\zeta)+(1-{\param})\check{Q}(\zeta).$$

Observe that when $\param=1$ we have $Q^{(\param)}=Q$ and if $\param=\infty$ we recover the hard edge potential $Q^{(\infty)}=Q+\infty\cdot \1_S$ which has been denoted $Q^S$ in papers such as \cite{HM,AKM,AKMW}.
As $\param\to 0$ we recover the obstacle function $Q^{(0)}=\check{Q}$, which is not an admissible potential since it fails to satisfy the growth condition \eqref{growth}.

Given a positive value of the confinement constant $\param$ it is natural to study corresponding planar Coulomb gas ensembles determined by a partition function of the form
$$Z_{n,\param}^{\beta}:=\int_{\C^n} e^{-\beta H_{n,\param}}\, dV_n,\qquad H_{n,\param}(\zeta_1,\ldots,\zeta_n)=\sum_{j\ne k}^n\log\frac 1 {|\zeta_j-\zeta_k|}
+n\sum_{j=1}^nQ^{(\param)}(\zeta_j),$$
where we write $dV_n$ for the usual Lebesgue measure in $\C^n$ divided by $\pi^n$. Here $\beta$ is an arbitrary positive constant
(the ``inverse temperature'').

We next introduce a Boltzmann-Gibbs type probability measure on $\C^n$ by
\begin{equation}\label{gibb}d\Prob_{n,\param}^{\beta}=\frac 1 {Z_{n,\param}^{\beta}}\cdot e^{-\beta H_{n,\param}}\, dV_n\end{equation}
and consider configurations $\{\zeta_j\}_1^n$ of points in $\C$, picked randomly with respect to this measure.

By arguing as in the free boundary case (see \cite{HM}) it is easy to verify that the system $\{\zeta_j\}_1^n$ roughly tends to follow the equilibrium distribution, in the sense that, for each bounded and continuous function $f$ on $\C$, one has the convergence
$$\frac 1 n \Expe_{n,\param}^{\beta}[f(\zeta_1)+\cdots +f(\zeta_n)]\to \sigma (f),\qquad (n\to\infty).$$

 As for any point process, the system $\{\zeta_j\}_1^n$ is determined by the collection of its $k$-point intensity functions $\bfR_{n,k}=\bfR_{n,k}^{(\param)}$. A characterization of these functions seems to be quite a hard enterprise. In this note, we shall henceforth restrict to the important \textit{determinantal case} $\beta=1$, leaving other $\beta$ to a future investigation.

In the case $\beta=1$, we have the basic determinant formula
$$\bfR_{n,k}^{}(\eta_1,\ldots,\eta_k)=\det(\bfK_n^{}(\eta_i,\eta_j))_{i,j=1}^k,$$
where the \textit{correlation kernel} $\bfK_n^{}$ can be taken as the reproducing kernel for the subspace $\calW_n^{}$ of $L^2=L^2(\C,dA)$ consisting of all ``weighted polynomials'' $w=pe^{-nQ^S/2}$ where $p$ is an analytic polynomial of degree at most $n-1$. (This \textit{canonical} correlation kernel is used
without exception below.)

We will write $\bfR_n^{}=\bfR_{n,1}^{}$ for the $1$-point function, which is the key player in our discussion below.

\subsection{Scaling limit}
Let us now fix a (regular) point on the outer boundary $\Gamma$, without loss of generality we place it at the origin, and so that the outwards normal to $\Gamma$ at $0$ points in the positive real direction.

We define a \textit{rescaled process} $\{z_j\}_1^n$ by magnifying distances about $0$ by a factor $\sqrt{n\Lap Q(0)}$,
$$z_j=\sqrt{n\Lap Q(0)}\,\zeta_j,\qquad\qquad (j=1,\ldots,n).$$

In general, we will denote by $\zeta,z$ two complex variables related by $z=\sqrt{n\Lap Q(0)}\, \zeta$.
We regard the droplet $S$ as a subset of the $\zeta$-plane.
 Restricting to a fixed bounded subset of the $z$-plane, the image of the droplet then more and more resembles left half plane $\L$, as $n\to\infty$; see Figure 2.

\begin{figure}[ht]
\begin{center}
\includegraphics[width=.4\textwidth]{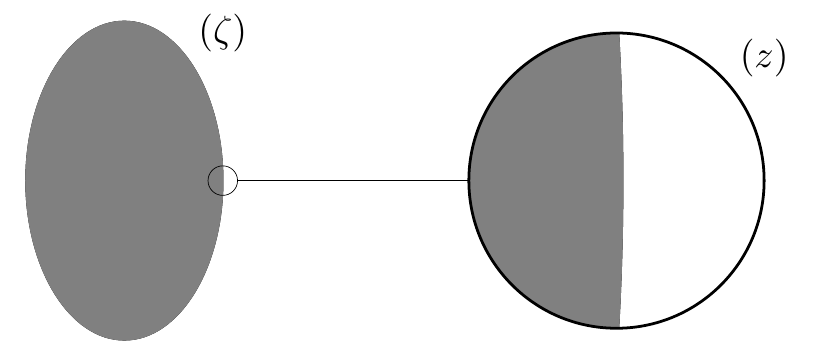}
\end{center}
\caption{Rescaling about a boundary point.}
\end{figure}

Following \cite{AKM,AKMW} we denote by plain symbols $R_n^{},K_n^{}$, etc., the $1$-point function, canonical correlation kernel, etc., with respect to the process $\{z_j\}_1^n$. Here the canonical kernel $K_n$ is, by definition
\begin{equation*}K_n^{}(z,w)=\frac 1 {n\Lap Q(0)}\bfK_n^{}(\zeta,\eta),\quad z=\sqrt{n\Lap Q(0)}\,\zeta,\quad w=\sqrt{n\Lap Q(0)}\,\eta.\end{equation*}

Recall that a function of the form $c(\zeta,\eta)=g(\zeta)\bar{g}(\eta)$, where $g$ is a continuous unimodular function, is called a \textit{cocycle}. A function $h(z,w)$ is \textit{Hermitian} if $\bar{h}(w,z)=h(z,w)$, and \textit{Hermitian-analytic} if furthermore $h$ is analytic in $z$ and in $\bar{w}$. Finally, the \textit{Ginibre kernel} is
$$G(z,w)=e^{-|z|^2/2-|w|^2/2+z\bar{w}}.$$
This is the correlation kernel of the infinite Ginibre ensemble, which emerges by rescaling about a regular bulk point, see e.g., \cite{AKM}.

\begin{thm}\label{exker} (``Structure of limiting kernels, Ward's equation'').
\begin{enumerate}[label=(\roman*)]
 \item \label{1o} There exists a sequence of cocycles $c_n$ such that each subsequence of the sequence $(c_nK_n)$ has a subsequence converging boundedly in $\C^2$ and
locally uniformly in $(\C\setminus i\R)^2$ to a Hermitian limit $K$.

\item\label{2o} Each limiting kernel $K$ in (i) is of the form
$$K(z,w)=G(z,w)\Psi(z,w)\exp\left\{(1-{\param})((\re z)_+^2+(\re w)_+^2)\right\}$$
where $\Psi$ is some Hermitian-entire function.

\item\label{3o} Each limiting 1-point function $R(z)=K(z,z)$ is everywhere strictly positive and satisfies the modified Ward equation
\begin{equation}\label{modward}\dbar C=R-1-\Lap\log R+(1-\param)\1_{\re z>0},\end{equation}
where $C(z)$ is the Cauchy-transform of the Berezin kernel corresponding to $R$ (see Section \ref{sec:Ward}).
\end{enumerate}
\end{thm}

By Theorem \ref{exker} and standard arguments (e.g., Macchi-Soshnikov's theorem, see \cite[Lemma 1]{AKMW}) we obtain immediately the existence and uniqueness 
of a non-trivial limiting point field
 $\{z_j\}_1^\infty$ corresponding to each limiting 1-point function $R$ in Theorem \ref{exker}.

 We have the following theorem.

\begin{thm}\label{thm:ml}
If $Q$ is radially symmetric, or more generally, if conditions (P1) and (P2) in Section \ref{apqq} are satisfied, then
the rescaled process $\{z_j\}_1^n$ converges as $n\to\infty$ to a unique determinantal point field $\{z_j\}_1^\infty$ determined by the $1$-point function $R^{(\param)}$ in \eqref{1pt}.
\end{thm}

Note in particular that Theorem \ref{thm:ml} proves existence of a determinantal point field with 1-point function $R^{(\param)}$.

Observe also that the $1$-point intensities $R^{(\param)}(z)$ converge locally uniformly as $\param\to0+$ to the function
\begin{align}R^{(0)}(z)&=\sfun_0(z+\bar{z})\exp\{2(\re z)_+^2\},\\
\sfun_0(z)&=\int_{-\infty}^0 \frac {e^{-(z-t)^2/2}}{\sqrt{2\pi}\fii(t)-t^{-1}e^{-t^2/2}}\, dt.
\end{align}

In view of this convergence, standard arguments imply that
 the point fields $\{z_j\}_1^\infty$ with 1-point functions $R^{(\param)}$ converge (in the sense of point fields) to a new determinantal point
field with 1-point function $R^{(0)}$. It is also easy to establish convergence on the level of Ward equations. We summarize this in the form of a theorem.

\begin{thm} There exists a unique determinantal point field in $\C$ with 1-point function $R^{(0)}$. The 1-point function $R^{(0)}$ gives rise to
a solution to the generalized Ward equation \eqref{modward} with parameter value $\param=0$.
\end{thm}

It is interesting to note that $R^{(0)}$ has a heavy tail in the sense that
\begin{align*}
R^{(0)}(x): &= \sfun_0(2x) \, e^{2x^2} = \frac{1}{4x^2} + O(x^{-3}), \quad x \to +\infty.
\end{align*}
By contrast, if $\param>0$ then $R^{(\param)}(x)\lesssim Ce^{-2\param x^2}$ as $x\to+\infty$, see Fig. 3.

\begin{figure}
\includegraphics[width=.45\textwidth]{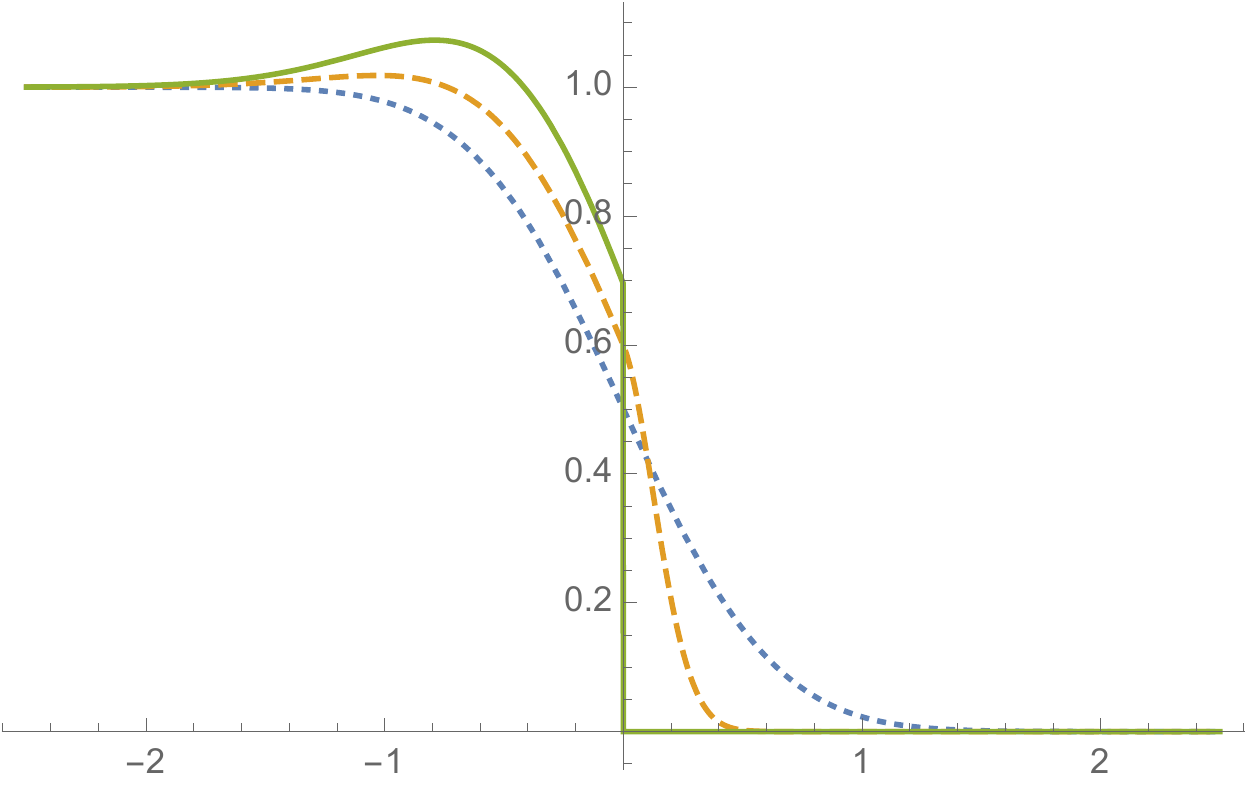}
\hspace{.05\textwidth}
\includegraphics[width=.45\textwidth]{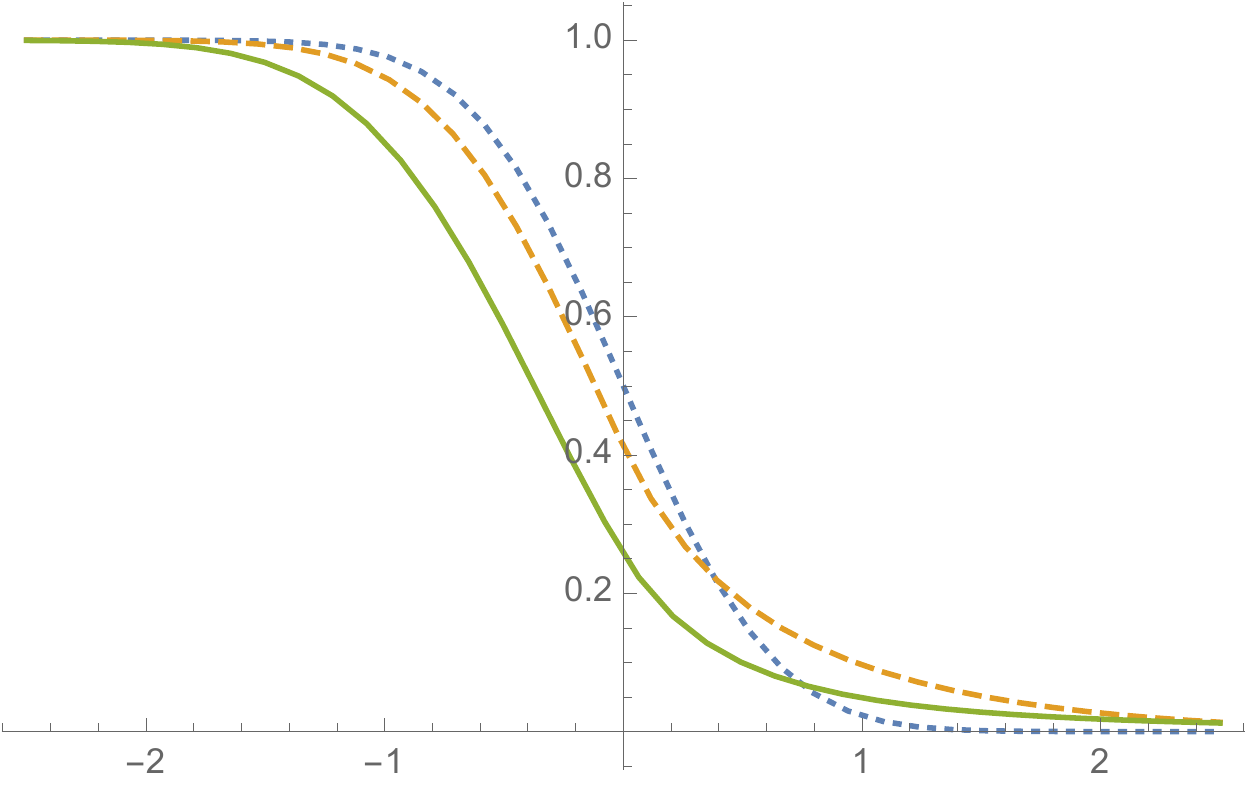}
\caption{Density profiles $R^{(\param)}$. The graphs for $\param=1$(blue dotted line), $10$(orange dashed line), and $\infty$(green line) are on the left, and those for $\param=1$(blue dotted line), $0.1$(orange dashed line), $0$(green line) are on the right.}
\label{fig3}
\end{figure}

\begin{rmk*} Likewise, by letting $\param\to+\infty$, we can recover the results concerning the hard edge point field with 1-point function $R^{(\infty)}$ from the note \cite{A18}. We refer to \cite{AKMW,Fo} for further details about hard edge point fields and their applications.
\end{rmk*}

\subsection{Distribution of the maximum modulus} Let us now assume that the external potential $Q$ is radially symmetric, and let $\{\zeta_j\}_1^n$ be a random sample from the corresponding Gibbs' distribution.

We shall denote by the symbol ``$|\zeta|_n$'' the maximum modulus of the sample, i.e., the random variable
$$|\zeta|_n = \max_{1\leq j\leq n}|\zeta_j|.$$
We also write
$$\rho=\max\{|\zeta|;\, \zeta\in S\}$$
and introduce the constants
 $$\gamma_n = \log (n/2\pi) -2\log\log n + 2\log C_{\param},$$
where $C_{\param} = \rho\sqrt{\Lap Q(\rho)}/\Phi_{\param}(0) = 2\rho \sqrt{{\param}\Lap Q(\rho)}/(\sqrt{\param}+1)$.

Finally, we define a random variable $\omega_n$ by rescaling $|\z|_n$ about $\rho$ in the following way:
$$\omega_n = \sqrt{4n{\param}\gamma_n\Lap Q(\rho)}\left(|\zeta|_n - \rho - \sqrt{\frac{\gamma_n}{4n{\param}\Lap Q(\rho)}}\right).$$

Given these proviso, we have the following theorem, which generalizes earlier results due to Rider \cite{R03} and Chafa\"{i} and P\'{e}ch\'{e} \cite{CP14} in the case $\param=1$.

\begin{thm}\label{thm:sr} The random variable $\omega_n$ converges in distribution to the standard Gumbel distribution: for $x\in \R$
$$\lim_{n\to \infty} \Prob_n (\omega_n \leq x)= e^{-e^{-x}}.$$
\end{thm}

\begin{rmk*} The Gumbel distribution is used to describe the distribution of extreme values. For example, the fluctuation of the maximum of i.i.d. gaussian random variables is expressed by a Gumbel distribution.
In the present context, it was shown by Rider \cite{R03} that the scaled maximal modulus of the free boundary Ginibre ensemble converges to the Gumbel distribution. This result was later generalized
to arbitrary radially symmetric potentials by Chafa\"{i} and P\'{e}ch\'{e} \cite{CP14}.

 For the hard edge ensemble, the situation is very different. In microscopic scale, the eigenvalues are distributed densely near the boundary inside the droplet. With a proper scaling, the limit law for the maximal modulus follows an exponential distribution. See \cite{Seo}. The exponential-type distribution and the Gumbel distribution can be found in the classification of extremal distribution functions. See \cite[Section 14]{Bl95}.

On the other hand, it is shown by Butez and Garc\'{i}a-Zelada in \cite{BGZ18} that for an $n$-dependent potential of the form $q_n:=(1+\frac{1}{n})\check{Q}$ (where $Q$ is a suitable radially symmetric potential for which the corresponding droplet is the unit disk) the maximal modulus $|\zeta|_n $ converges in distribution (without scaling) to a random variable ``$|\zeta|_\infty$'' with distribution function
$$\Prob(|\zeta|_\infty < x) = \prod_{k=1}^{\infty}\left(1-x^{-2k}\right),\quad x>1.$$
The potential $q_n$ is very weakly confining, and is of another type from the ones studied in the present paper.
\end{rmk*}

\subsection{Plan of this paper; further results}
The outline of this paper is as follows. In Section \ref{sec:ex}, we give an elementary proof of Theorem \ref{thm:ml} for the generalized Ginibre ensemble (the potential $Q(\zeta)=|\zeta|^2$)
with an arbitrary value of the confinement coefficient $\param$.

In Section \ref{sec:struct}, we prove Theorem \ref{exker} about the structure of limiting kernels.

In Section \ref{sec:Ward} we analyze the modified Ward equation \eqref{modward} under some natural assumptions, most importantly we assume apriori
\textit{translation invariance} of a scaling limit, i.e., $R(x+iy)=R(x)$. Under this assumption,
it turns out that Ward's equation determines the solution $R=R^{(\param)}$ up to constant. (See Theorem \ref{ward1}.) This argument applies
for any reasonable potential, and gives a possible ``abstract'' approach to universality.

In Section \ref{prel} and \ref{apqq} we present a more ``concrete'' approach, by adapting the method of quasipolynomials of Hedenmalm and Wennman \cite{HW17}.
This method is then applied to prove universality of scaling limits (Theorem \ref{thm:ml}) in Section \ref{sec:err} and universality of maximum modulus (Theorem \ref{thm:sr}) in Section \ref{sec:max}.

In the concluding remarks section (Section \ref{sec:sum}) we summarize our results and mention some natural problems going forward.

\subsection{Index of notation}  We will write $L^2(e^{-\phi})=L^2(e^{-\phi},\C)$ for the usual $L^2$-space normed by $\|f\|_\phi^2:=\int_\C|f|^2 e^{-\phi}\, dA$; when $\phi=0$ we drop the  ``$\phi$'' and write $L^2$ and $\|f\|$ respectively. The characteristic function of a set $E$ is denoted
$\1_E$, and we always reserve the symbol ``$\delta_n$'' for the number $\delta_n=n^{-1/2}\log n$. For convenience of the reader, we now list a number of other frequently occurring symbols.

$\hat{\C}=\C\cup\{\infty\}$; $\D=\{|\zeta|<1\}$; $\D_e(\rho)=\{|\zeta|>\rho\}\cup\{\infty\}$; $\D_e=\D_e(1)$; $\L=\{\re z<0\}$;

``Area'': $dA(\zeta)=\frac 1 \pi \, d^2 \zeta$; ``Arclength'': $ds(\zeta)=\frac 1 {2\pi}\,|d\zeta|$; ``Laplacian'': $\Lap =\d\dbar$;

$S_\tau$: droplet in external potential $Q/\tau$; $U_\tau=\hat{\C}\setminus\Pc S_\tau$; $\Gamma_\tau=\d U_\tau$;

$\normal_\tau$: exterior unit normal on $\Gamma_\tau$;

$\check{Q}_\tau$: obstacle function in external potential $Q/\tau$;

$V_\tau$: harmonic continuation of $\check{Q}_\tau|_{U_\tau}$ across $\Gamma_\tau$;

$\phi_\tau$: normalized conformal map $\phi_\tau : U_\tau\to\D_e$ with $\phi_\tau(\infty)=\infty$ and $\phi_\tau'(\infty)>0$.

\section{Confined Ginibre ensembles}\label{sec:ex}

In this section, we consider the Ginibre case, i.e., the random normal matrix model associated with the external potential
$$Q^{(\param)} = {\param}Q + (1-{\param})\check{Q},\quad ({\param}>0)$$
where $Q(\z) = |\z|^2.$
In this example, the droplet $S$ is the unit disk $\D=\{\zeta\in\C: |\zeta| \leq 1\}$ and
the solution to the obstacle problem is given by
$$\check Q(\z) = \begin{cases} |\z|^2 \quad  &\textrm{if } |\z|\le 1, \\  1 + \log |\z|^2  &\textrm{otherwise.}\end{cases}$$
Let $\{\zeta_j\}_1^n$ denote the eigenvalues of random normal matrices with $Q^{(\param)}$. Define a rescaled system $\{z_j\}_1^n$ at the boundary point $p=1$ by
$z_j = \sqrt{n}(\zeta_j -1)$.

We shall give a short proof for the convergence of the rescaled ensemble to the point field in Theorem \ref{thm:ml}. We here use the normal approximation to the Poisson distribution. This approach can be found in \cite{AKM} for the Ginibre ensemble with a free boundary as well as with a hard edge.

\begin{proof}[Proof of Theorem \ref{thm:ml} for the Ginibre case]
Recall that a correlation kernel for the process $\{\zeta_j\}_1^n$ is obtained by
$$\mathbf{K}_n(\zeta,\eta) = \sum_{j=0}^{n-1}\frac{(\z\bar{\w})^j}{\|\z^j\|^2_{nQ^{(\param)}}} e^{-n(Q^{(\param)}(\zeta)+Q^{(\param)}(\eta))/2}$$
since $Q^{(\param)}$ is radially symmetric.
Here, by direct computation we have
\begin{align*}
\|\zeta^j\|^2_{nQ^{(\param)}}= n^{-j-1}\gamma(j+1,n) + e^{-n(1-{\param})}(n{\param})^{-j-1 + n(1-{\param})}\Gamma(j+1-n(1-{\param}),n{\param}),
\end{align*}
where $\gamma(j+1,n) = \int_{0}^{n} s^j e^{-s} ds$ is the lower incomplete Gamma function and $\Gamma(j+1,n)=\int_{n}^{\infty} s^j e^{-s} ds$ is the upper incomplete Gamma function.
It follows that
$$\mathbf{K}_n(\zeta,\eta) = n \sum_{j=0}^{n-1}\frac{(n\zeta\bar\eta)^j e^{-nQ^{(\param)}(\zeta)/2-nQ^{(\param)}(\eta)/2}}{P(j,n,{\param})},$$
where
\begin{align*}
P(j,n,{\param}) &=  j! \, \Prob(U_n>j) \\
&+ e^{-n(1-{\param})}  {\param}^{-j-1} (n{\param})^{n(1-{\param})}\, {\Gamma(j+1-n(1-{\param}))}\, \Prob(U_{n{\param}}\le j+1-n(1-{\param}))
\end{align*}
and $U_s \sim \text{Po}(s)$, i.e., $U_s$ is a Poisson distributed random variable with intensity $s$.

By normal approximation of the Poisson distribution we obtain
$$\Prob(U_n>j)  =  \varphi\big(\xi_{j,n}\big) \big(1+o(1)\big), \quad \Prob(U_{n{\param}}\le j+1-n(1-{\param}))= 1 - \varphi\left(\frac{\xi_{j,n}}{\sqrt{\param}}\right) \big(1+o(1)\big),$$
where $\fii(\xi)=\frac 1 {\sqrt{2\pi}}\int_\xi^\infty e^{-t^2/2}\, dt,$ $\xi_{j,n}=\frac{(j-n)}{\sqrt{n}},$ and $o(1)\to 0$ as $n\to\infty$ uniformly in $j$.

Now rescale by
$$\zeta=1+z/\sqrt{n},\qquad \eta=1+w/\sqrt{n},\qquad z,w\in\C$$
and write
$$K_n(z,w)=\frac 1 n \mathbf{K}_n(\zeta,\eta)=\sum_{j=0}^{n-1}\Big( \frac {n\zeta\bar{\eta}} \lambda\Big)^j\frac {\lambda^j}{P(j,n,{\param}) }e^{-\lambda},$$
where $\lambda=n\big(Q^{(\param)}(\zeta)+Q^{(\param)}(\eta)\big)/2.$
Note that $\lambda$ has the asymptotic identity:
\begin{align*}
\lambda &= n + \sqrt n \re(z+w) + \frac{|z|^2+|w|^2}2 - (1-{\param})\Big( (\re z)_+^2 + (\re w)_+^2\Big) + o(1).
\end{align*}
We have shown that
\begin{equation*}K_n(z,w)=(1+o(1))\sum_{j=0}^{n-1}\left( \frac {n\zeta\bar{\eta}} \lambda\right)^j\frac {\lambda^j}{j!}e^{-\lambda}\frac 1 {\varphi(\xi_{j,n}) + \Gamma(j,n,{\param})(1-\varphi(\xi_{j,n}/\sqrt{\param}))},\end{equation*}
where $\Gamma(s,n,{\param}):=e^{-n(1-{\param})}  {\param}^{-s-1} (n{\param})^{n(1-{\param})} \Gamma(s+1-n(1-{\param}))/ \Gamma(s+1).$
Finally, if $X_n\sim \text{Po}(\lambda)$, we can write the last sum in the form
\begin{align*}
&\sum_{j=0}^{n-1}
\Big(\frac {n\zeta\bar{\eta}}\lambda\Big)^j \frac {\Prob(X_n=j)} {\varphi(\xi_{j,n}) + \Gamma(j,n,{\param})(1-\varphi(\xi_{j,n}/\sqrt{\param}))} \\
&=\mathbb{E} \, \Big(\frac {n\zeta\bar{\eta}}\lambda\Big)^{X_n}
\frac {\1_{\{X_n<n\}}} {\varphi\Big(\frac {X_n-n}{\sqrt{n}}\Big)+\Gamma(X_n,n,{\param})\Big(1-\varphi\Big(\frac {X_n-n}{\sqrt{n{\param}}}\Big)\Big)}.
\end{align*}
Defining $Y_n$ by $X_n=\lambda+\sqrt{\lambda}Y_n$ and $\alpha_n=(n-\lambda)/\sqrt{n},$ we now get a relation of the form
$$K_n(z,w)=(1+o(1))A_n  B_n,$$
where
$$A_n = \Big(\frac {n\zeta\bar{\eta}}\lambda\Big)^\lambda$$
and
$$B_n =
\mathbb{E} \, \Big(\frac {n\zeta\bar{\eta}}\lambda\Big)^{\!\sqrt{\lambda} Y_n}\!\!\!
\frac {\1_{\{Y_n<\alpha_n\}}} {\varphi\Big(\sqrt{\frac{\lambda}{n}}\,Y_n\,-\alpha_n\Big)+\Gamma(\lambda+\sqrt{\lambda}Y_n,n,{\param})\Big(1-\varphi\Big(\sqrt{\frac{\lambda}{n{\param}}}\,Y_n\,- \frac{\alpha_n}{\sqrt{\param}}\Big)\Big)}.$$

\begin{lem} \label{lem:anbn}
Let $x = \re z$, $u = \re w$, and $b = b(z,w) = \im(z+\bar{w})$.
We have
\begin{align}\label{an}
A_n &= e^{ib\sqrt n}  e^{b^2/2} G(z,w)\big(1+o(1)\big) \, e^{(1-{\param})(x_+^2 + u_+^2)}, \\ \label{bn}
B_n &= e^{-b^2/2} \big(1+o(1)\big) S(z+\bar w),
\end{align}
where $G(z,w) = e^{z\bar w -|z|^2/2-|w|^2/2}$ and
$$S(z) = \frac 1 {\sqrt{2\pi}}\int_{-\infty}^0
{e^{-(\xi-z)^2/2}} \frac{1}{\Phi_{\param}(\xi)}\, d\xi,
\qquad
\Phi_{\param}(\xi) = \varphi(\xi) +  \frac1{ \sqrt{\param} } e^{\frac{1-{\param}}{2{\param}} \xi^2}\left(1-\varphi\left(\frac \xi{\sqrt{\param}}\right)\right).$$
\end{lem}
\begin{proof}
We first analyze that case $\zeta = \eta = 1 + x/\sqrt n$ is real.
We find
$$ \lambda =
 \begin{cases}
 n + 2x\sqrt n + x^2 \quad &\textrm{if }x\le0, \\
 n + 2x\sqrt n + (2{\param}-1)x^2 + o(1) &\textrm{otherwise,}
\end{cases}$$
and
$$A_n =
\begin{cases}
1 \quad &\textrm{if }x\le0, \\
 e^{2(1-{\param})x^2}(1+o(1)) &\textrm{otherwise.}
 \end{cases}$$
Using $\alpha_n =  -2x + o(1),$ $\lambda/n = 1 + o(1),$ and Stirling's formula, we obtain the asymptotic identity for $\Gamma(\lambda+\sqrt{\lambda}s,n,{\param}):$
$$\Gamma(\lambda+\sqrt{\lambda}s,n,{\param}) =  \frac1{ \sqrt{\param} } e^{\frac{1-{\param}}{2{\param}} (s+2x)^2}+ o(1).$$
By means of the central limit theorem and the above asymptotic identities, we now approximate the factor $B_n$ as follows,
\begin{align*}B_n&= \frac 1 {\sqrt{2\pi}}\int_{-\infty}^{-2x}
\frac {e^{-s^2/2}\, ds} {\varphi(s+2x) +  \frac1{ \sqrt {\param} } e^{\frac{1-{\param}}{2{\param}} (s+2x)^2}(1-\varphi((s+2x)/\sqrt{\param})) } + o(1).
\end{align*}
Now we consider $A_n$ and $B_n$ for $$\zeta=1+z/\sqrt{n},\qquad \eta=1+w/\sqrt{n},\qquad z,w\in\C$$
and recall that
\begin{align*}
\lambda = n + \sqrt n \re(z+w) + \frac{|z|^2+|w|^2}2
- (1-{\param})\left( x_+^2  + u_+^2\right) + o(1).
\end{align*}
We use the same argument as in the previous case and deduce \eqref{an} and \eqref{bn}.
\end{proof}

By the above lemma, we approximate the kernel $K_n(z,w)$ by
$$K_n(z,w) =  G(z,w) e^{(1-{\param})(x_+^2 + u_+^2)}S(z+\bar w)\big(1+o(1)\big)$$
up to cocycles and the proof is complete.
\end{proof}

\section{Structure of limiting correlation kernels} \label{sec:struct}

In this section we prove Theorem \ref{exker}. Our proof closely follows the derivation of similar results in \cite{AKM}; we shall
here indicate the necessary modifications.

We start by proving some generally useful estimates for functions of the form
$$w(\zeta)=p(\zeta)e^{-n Q^{(\param)}(\zeta)/2},$$
where $p$ is a holomorphic function (in practice, a polynomial or a ``quasi-polynomial'' - see below). 

 It is convenient to introduce the numbers
$$M=M_{n,\param}:=\|\Lap Q^{(\param)}\|_{L^\infty(S_n)},\quad (S_n:=\{\zeta;\, \dist(\zeta,S)\le 1/\sqrt{n}\}).$$ The
numbers $M_{n,\param}$ clearly do not increase
as $n$ increases; in the sequel we may fix $M$ to be $M_{n_0,\param}$ with a sufficiently large $n_0$.

\begin{lem} \label{propp} Suppose that $\zeta\in S$ and $n$ is large enough and suppose that $p$ is holomorphic in the disc $D(\zeta;1/\sqrt{n})$. Then there is a constant $C=C_\param$ depending only on $M$ such that $$|w(\zeta)|^2\le Cn\int_{D(\zeta;1/\sqrt{n})}|w|^2\, dA.$$
\end{lem}

\begin{proof} This follows from a standard argument, using that $z \mapsto |w(\zeta+z/\sqrt{n})|^2 e^{M|z|^2}$ is
logarithmically subharmonic in $D(0;1)$.
(See for instance \cite[Section 3]{AKM}.)
\end{proof}

\begin{lem} \label{l6} Let $K$ be a compact subset of the interior of $\Pc S$ and let $w=qe^{-nQ^{(\param)}/2}$ where $q$ is holomorphic
on $\C\setminus K$ and satisfying $q(\zeta)=O(\zeta^j)$ as $\zeta\to\infty$. Write $\tau=j/n$ and suppose that $\tau$ is in the range $\tau\le 1$. Then there is a constant $C_\param>0$ such that
$$|w(\zeta)|\le C_\param\sqrt{n}\|w\|e^{-n(Q^{(\param)}-\eqpot_\tau)(\zeta)/2},\qquad (\dist(\zeta,K)\ge 1/\sqrt{n}).$$
\end{lem}

\begin{proof}
By Lemma \ref{propp}
there is a constant $C$ such that $|w(\zeta)|\le \sqrt{C n}\|w\|$ when $\zeta\in S_\tau$.

Consider the subharmonic function
\begin{equation*}\label{smock}u_n(\zeta)=\frac 1 n\log\left(\frac {|q(\zeta)|^2}{C n\|w\|^2}\right).\end{equation*}
Note that $u_n(\zeta)\le \tau\log|\zeta|^2+O(1)$ as $\zeta\to\infty$ and $u_n\le Q$ on $S_\tau$, and $u_n$ is subharmonic on $\C$. Hence $u_n\le\eqpot_\tau$ on $\C$, proving the lemma
(with $C_\param:=\sqrt{C}$).
\end{proof}

\begin{rmk*} Following \cite{HW17}, we call $q$ above a quasipolynomial of degree $j$; $w$ is a weighted quasipolynomial.
\end{rmk*}

\begin{lem} \label{goodb} The $1$-point function $\bfR_n$ in external potential $Q^{(\param)}$ satisfies
$$\bfR_n(\zeta)\le C_\param ne^{-n\param(Q-\eqpot)(\zeta)},\qquad (\zeta\in \C).$$
\end{lem}

\begin{proof} Let $\bfk_n$ be the reproducing kernel of the space of holomorphic polynomials of degree at most $n-1$ with the norm of
$L^2(e^{-nQ^{(\param)}})$. Now fix $\zeta\in S$ and put $q(\eta)=\bfk_n(\eta,\zeta)/\sqrt{\bfk_n(\zeta,\zeta)}$ and $w=qe^{-nQ^{(\param)}/2}$. Then $\|w\|=1$, and so by Lemma \ref{l6}
(with $\tau=1$) we have
$$\bfR_n(\zeta)=|w(\zeta)|^2\le C_\param ne^{-n(Q^{(\param)}-\eqpot)(\zeta)}=C_\param ne^{-n\param(Q-\eqpot)(\zeta)}.$$
\end{proof}

Fix a neighborhood $\Omega$ of $S$ where $Q$ is real-analytic and strictly subharmonic. Let $A$ be a Hermitian-analytic function defined in a neighborhood in $\C^2$ of the set $\{(\zeta,\zeta):\zeta\in \Omega\}$ such that $A(\zeta,\zeta)=Q(\zeta)$.
As in \cite{AKM}, we fix an outer boundary point $0\in \Gamma=\d\Pc S$ such that the outwards normal to $\Gamma$ at $0$ points in the positive real direction.

We define a kernel by
$$\mathbf{k}_n^{\#}(\zeta,\eta)=n(\pa_1\bar\pa_2A)(\zeta,\eta)\, e^{nA(\zeta,\eta)}.$$
We also introduce rescaled kernels $k_n^{\#}$ and $K_n^{\#}$ by
$$k_n^{\#}(z,w)=\frac{1}{n\Lap Q(0)}\mathbf{k}_n^{\#}(\zeta,\eta);\quad K_n^{\#}(z,w)=k_n^{\#}(z,w)e^{-n(Q^{(\param)}(\zeta)+Q^{(\param)}(\eta))/2}$$
where $\zeta = z/\sqrt{n\Lap Q(0)}$ and
$\eta = w/\sqrt{n\Lap Q(0)}$.

\begin{lem} For $z$ in a given compact subset of $\C$, we have as $n\to\infty$
\begin{equation}\label{taye}
(Q-\check{Q})(\zeta) = 2n^{-1}(\re z)_+^2(1+O(n^{-1/2})).\end{equation}
Moreover, if $R_n(z)=\frac 1 {n\Lap Q(0)}\bfR_n(\zeta)$ denotes the rescaled 1-point function, then
\begin{equation}\label{taye2}R_n(z)\le C_\param e^{-2\param(\re z)_+^2}.\end{equation}
\end{lem}

\begin{proof} Let $V$ denote the harmonic continuation of $\check{Q}|_{S^c}$ across the outer boundary curve $\Gamma=\d \Pc S$ to a neighborhood $N$ of $0$.
Let $\pa_\rmn$ and $\pa_\rmt$ denote the exterior normal and tangential derivatives on $\Gamma\cap N$, respectively.
Since $Q=V$ and $\nabla Q = \nabla V$ on $N\cap\Gamma$, we have
$$\pa_\rmn(Q-V)= \pa_\rmt(Q-V) = \pa_\rmt^2(Q-V) =0.$$
This implies $\pa_\rmn^2(Q-V)=4\Lap(Q-V) = 4\Lap Q$ on $\Gamma$, proving \eqref{taye}.
To prove \eqref{taye2} it now suffices to appeal to Lemma \ref{goodb}.
\end{proof}

\begin{proof}[Proof of Theorem \ref{exker}]
Define a function $\Psi_n$ by
$$\Psi_n(z,w) = {K_n(z,w)}/{K_n^{\#}(z,w)}.$$
Then $\Psi_n$ is Hermitian-analytic and the rescaled kernel $K_n$ can be expressed as
\begin{equation*}K_n(z,w)=k_n^{\#}(z,w)\Psi_n(z,w) e^{-n(Q^{(\param)}(\zeta)+Q^{(\param)}(\eta))/2}.\end{equation*}
Using the argument in \cite[Lemma 3.4]{AKM}, we see that there exist cocycles $c_n$ such that
\begin{equation}\label{cycon}
c_n(z,w)\, k_n^{\#}(z,w)\,e^{-n(Q(\zeta)+Q(\eta))/2}=G(z,w)(1+o(1))
\end{equation}
where $o(1)\to 0$ locally uniformly on $\C^2$ as $n\to \infty$. Indeed, \eqref{cycon} is obtained by the Taylor expansion (assuming $\rmn =1$ for simplicity)
\begin{align*}& n(A(\zeta,\eta)-Q(\zeta)/2-Q(\eta)/2) = n(A(\zeta,\eta)-A(\eta,\eta))/2 - n(\overline{A(\zeta,\zeta)}-\overline{A(\eta,\zeta)})/2\\
&=i\im\Big[\frac{\sqrt{n}\,\pa_1A(0,0)}{\sqrt{\Lap Q(0)}}(z-w)+\frac{\pa_1^2A(0,0)}{2\Lap Q(0)}(z^2-w^2)\Big]+ z\bar{w} -|z|^2/2 -|w|^2/2.
\end{align*}
Since $Q^{(\param)}(\zeta) = \check{Q}(\zeta) + \param(Q-\check{Q})(\zeta),$
we obtain the convergence
\begin{equation}\label{Kns}
c_n(z,w)\, K_n^{\#}(z,w)=G(z,w)e^{(1-{\param})((\re z)_+^2+(\re w)_+^2)}(1+o(1))
\end{equation}
by \eqref{taye} and \eqref{cycon}. Here, $o(1) \to 0$ in $L^1_{\loc}(\C^2)$ and locally uniformly on $(\C\setminus i\R)^2$.

On the other hand, observe that
$$|\Psi_n(z,w)|^2 = \left|\frac{K_n(z,w)}{K_n^{\#}(z,w)}\right|^2 \leq \frac{R_n(z)R_n(w)}{|c_n(z,w)K_n^{\#}(z,w)|^2}.$$
Combining with \eqref{taye2},
we obtain that the family $\{\Psi_n\}$ is locally bounded on $\C^2$. By the normal family argument in \cite{AKM}, for each subsequence of $\{\Psi_n\}$ there exists a further subsequence which converges locally uniformly on $\C^2$ and every limit $\Psi = \lim_k \Psi_{n_k}$ is Hermitian-entire. This proves part \ref{1o} and \ref{2o} of Theorem \ref{exker}.

 For the remaining part, we observe that the modified Ward equation \eqref{modward} follows by applying the rescaling procedure in \cite{AKM} to the potential $Q^{(\param)}$ in a straightforward way. Moreover, to prove non-triviality using the method in \cite{AKM} we need only to prove that $R(z)>0$ at some point $z$. For $z$ somewhat inside the (rescaled) droplet
this can be done exactly as in \cite{AKM}, since $Q=Q^{(\param)}$ in $S$.
\end{proof}

\section{An approach to universality using Ward's equation}\label{sec:Ward}
In this section, we shall analyze solutions
$R(z)$ to the (modified) Ward equation \eqref{modward} which satisfy the ``physically reasonable'' condition of translation invariance. This leads to an approach to
universality of the 1-point function $R^{(\param)}$ based only on Ward's equation.

Given the nature of our rescaling (with the rescaled droplet occupying the left half plane) it is natural to assume that each limiting
1-point function $R$ in Theorem \ref{exker} should be vertical translation invariant, in the sense that $R(z)=R(z+it)$ for each $t\in\R$.
This is equivalent to that the holomorphic kernel $\Psi$ should satisfy $\Psi(z,w)=S(z+\bar{w})$ for some entire function $S$. (In view of the computations in the preceding section, we know that
this holds for the Ginibre ensemble, in which case $S=\sfun_\param$.)

In the following, we shall work within the class of abstract translation invariant ``limiting kernels''
where
$S$ has the ``error function-type'' (cf. \cite{AKMW})
$$S(z) = \frac 1 {\sqrt{2\pi}}\int_{-\infty}^{+\infty} {e^{-(\xi-z)^2/2}} {s(\xi)}\, d\xi$$
for some Borel measurable, non-negative function $s$ on $\R$.

To each such function we define a corresponding function $B=B_S$
(the ``Berezin kernel'') by
$$B(z,w) = e^{-|z-w|^2} e^{2(1-{\param})(\re w)_+^2} \frac{S(z+\bar w)S(\bar z + w)}{S(z+\bar z)}.$$
We will describe all the functions $s$ such that the Berezin kernel $B$ satisfies two conditions called the mass-one condition and Ward's equation.

We say that $B$ satisfies the \textit{mass-one equation} if for all $z$
\begin{equation}\label{mass}
\int_{\mathbb{C}} B(z,w) \,dA(w) = 1.
\end{equation}

In the case of a regular point on the boundary, \textit{Ward's equation} has the form
\begin{equation}\label{ward}
\bp C= R -1 - \Lap\log R + (1-{\param})\1_{\re z>0},\end{equation}
where
$R(z) = B(z,z)$ and
\begin{equation}\label{cauc}C(z)=\int_\mathbb{C}\frac{B(z,w)}{z-w}\,dA(w).\end{equation}

Throughout this section, we write $z=x+iy$ and $w=u+iv$ with $x,y,u,v\in \R$.
\begin{thm}\label{mass1}
The mass-one condition \eqref{mass} holds if and only if
$$s(\xi) = \frac{\mathbf{1}_E}{\Phi_{\param}(\xi)}$$
almost everywhere for some Borel subset $E\subset\R$, where $\Phi_\param:\R\to (0,\infty)$ is the function defined in \eqref{defphi}.
\end{thm}

\begin{proof}
A straightforward computation shows that
$$
\sqrt{\frac2\pi} \int_{-\infty}^\infty  e^{-2u^2 + 2\xi u+2(1-{\param})u_+^2}\,du = e^{\frac12\xi^2}\varphi(\xi)+\frac1{\sqrt {\param}}\,e^{\frac1{2{\param}}\xi^2}\left(1-\varphi\left(\frac{\xi}{\sqrt{\param}}\right)\right).$$
It follows immediately that
\begin{align*}
\int_{\mathbb{C}}&e^{-|z-w|^2} e^{2(1-{\param})u_+^2} S(z+\bar w) S(\bar z + w) \,dA(w) \\
&= \frac1\pi \int_{-\infty}^\infty  \int_{-\infty}^\infty {e^{-2x^2-2u^2-\xi^2 + 2\xi(x+u)+2(1-{\param})u_+^2}} {s(\xi)^2}\,du d\xi \\
&= \frac1{\sqrt{2\pi}}\int_{-\infty}^\infty {e^{-\frac12(\xi-2x)^2}}{s(\xi)^2}\left(\varphi(\xi) +  \frac1{ \sqrt{\param} } e^{\frac{1-{\param}}{2{\param}} \xi^2}\left(1-\varphi\left(\frac \xi{\sqrt{\param}}\right)\right)\right)\,d\xi.
\end{align*}
Here, the expression in the large parenthesis is recognized as the function $\Phi_\param$ defined in \eqref{defphi}.

Equating the last expression with $S(2x),$ we see that the mass-one equation is equivalent to that (for a.e.-$\xi$)
$$ s(\xi) = {s(\xi)^2}\,\Phi_{\param}(\xi).$$
The last condition means that (up to sets of measure $0$) $s(\xi)=1/\Phi_\param(\xi)$ whenever $s(\xi)\ne 0$.
\end{proof}

We will now analyze Ward's equation \eqref{ward}. We will require a few initial remarks about the Cauchy transform \eqref{cauc}.

By a change of variables, the Cauchy transform $C$ takes the form
$$C(z) = - \int_{-\infty}^\infty \int_{-\infty}^\infty \frac{e^{-u^2-v^2+2(1-{\param})(u+x)_+^2}}{u+iv} \frac{S(2x+u-iv) S(2x+u+iv)}{S(2x)} \,\frac{dv\,du}\pi.$$
Let
$$
L(x) = \int_{-\infty}^\infty\int_{-\infty}^\infty \frac{e^{-u^2-v^2+2(1-{\param})(u+x/2)_+^2}}{u+iv} S(x+u-iv) S(x+u+iv) \,\frac{dv\,du}\pi.
$$
Then
$$C(z) = -\frac{L(2x)}{S(2x)},\qquad \bp C(z) = - \Big(\frac LS\Big)'(2x).$$
Now we describe all solutions to Ward's equation \eqref{ward}.
\begin{thm}\label{ward1}
Ward's equation holds if and only if
$$s(\xi) = \frac{\1_I}{\Phi_{\param}(\xi)}
$$
almost everywhere, where $\Phi_\param$ is the function \eqref{defphi} and
$I$ is an interval in $\R$. Furthermore, if $R$ satisfies \begin{equation}\label{Rleft} \liminf_{x\to -\infty} R(x) > 0,\end{equation}
then $I = (-\infty, c)$ for some constant $c$.
\end{thm}

\begin{proof}
Let
$$
\ell(x;\xi,\eta)
=\int_{-\infty}^\infty\int_{-\infty}^\infty  l(x;\xi,\eta;u,v)\,\frac{e^{-iv(\xi-\eta)}}{u+iv}\,\frac{dv\,du}\pi,$$
where
$$l(x;\xi,\eta;u,v)=e^{-2u^2 -\frac12(\xi-x)^2-\frac12(\eta-x)^2+u(\xi+\eta-2x)+2(1-{\param})(u+x/2)_+^2}$$
so that
$$L(x) = \int_{-\infty}^\infty \int_{-\infty}^\infty {\ell(x;\xi,\eta)}{s(\xi)s(\eta)}\,\frac{d\xi\,d\eta}{2\pi}.$$
First, observe that
\begin{equation} \label{eq: lemma}
\int_{-\infty}^\infty \frac{e^{-iv(\xi-\eta)}}{u+iv}\,\frac{dv}\pi =
\begin{cases}
-2  e^{u\cdot(\xi-\eta)}\mathbf{1}_{\xi>\eta},  &\textrm{if } u < 0;\\
\phantom{-}2 e^{u\cdot(\xi-\eta)}\mathbf{1}_{\xi<\eta}, &\textrm{if } u > 0.
\end{cases}
\end{equation}
Let
$$\ell_1(x;\xi,\eta)=- \1_{\xi>\eta} \int_{-\infty}^\infty 2 e^{-2u^2 -\frac12(\xi-x)^2-\frac12(\eta-x)^2+2u(\xi-x)+2(1-{\param})(u+x/2)_+^2}\,du$$
and
$$L_1(x) = \int_{-\infty}^\infty \int_{-\infty}^\infty {\ell_1(x;\xi,\eta)}{s(\xi)s(\eta)}\,\frac{d\xi\,d\eta}{2\pi}.$$
By direct computation, we have
$$\ell_1(x;\xi,\eta) = -e^{-\frac12(\eta-x)^2} \sqrt{2\pi}  \1_{\xi>\eta} \Phi_{\param}(\xi)$$
and
\begin{equation}\label{eq: L1}
L_1(x) =  -\int_{-\infty}^\infty \int_{-\infty}^\infty   \1_{\xi>\eta} {e^{-\frac12(\eta-x)^2}}{s(\xi)s(\eta)} \Phi_{\param}(\xi) \,\frac{d\xi d\eta}{\sqrt{2\pi}}.
\end{equation}
Next, let  $L_2(x)=L(x)-L_1(x)$.
Then by \eqref{eq: lemma}, we have $$L_2(x)=\int_{-\infty}^\infty \int_{-\infty}^\infty {\ell_2(x;\xi,\eta)}{s(\xi)s(\eta)}\,\frac{d\xi\,d\eta}{2\pi},$$ where
$$\ell_2(x;\xi,\eta)
=  2\int_0^\infty  e^{-2u^2 -\frac12(\xi-x)^2-\frac12(\eta-x)^2+2u(\xi-x)+2(1-{\param})(u+x/2)_+^2}\,du.$$
This implies that
\begin{equation} \label{eq: L2}
L_2(x)
= S(x) \int_x^\infty S(v)e^{\frac12(1-{\param})v_+^2}\,dv.
\end{equation}
From the relation
$$R(z) = e^{2(1-{\param})x_+^2} S(z+\bar z),$$
we have
$$\pa\bp \log R = (1-{\param})\mathbf{1}_{x>0} + \Big(\frac{S'}S\Big)'(2x).$$
Thus Ward's equation holds if and only if
$$\Big(\frac LS\Big)'(2x) =   \Big(\frac{S'}S\Big)'(2x) +1 - R(z).$$
From the equation~\eqref{eq: L2}, $L_2$ satisfies
$$ \Big(\frac {L_2}S\Big)'(2x) =  - R(z),$$ and hence
Ward's equation holds if and only if
$$\Big(\frac {L_1}S\Big)' =   \Big(\frac{S'}S\Big)' +1.$$
Equivalently,
$$L_1(x) = S'(x) +(x - c)S(x)$$
for some constant $c.$
We find
$$L_1(x) =  \int_{-\infty}^\infty (\eta - c) {e^{-\frac12(\eta-x)^2}}{s(\eta)} \,\frac{d\eta }{\sqrt{2\pi}}.$$
Comparing it to \eqref{eq: L1}, we have 
$$  \eta - c  = -\int_\eta^\infty s(\xi)\Phi_{\param}(\xi)\,d\xi$$
for almost everywhere in the set $I=\{\eta\in\R : s(\eta) \ne 0\}$.  This implies that $s \cdot \Phi_{\param}=\1_{I}$ almost everywhere and $I$ is connected. 
For the last statement, assume on the contrary that $I$ is left-bounded. Then it is a direct consequence that $R(x)\to 0 $ as $x\to -\infty$. This implies that if $R$ satisfies the condition \eqref{Rleft}, then $I=(-\infty, c)$ where  
$c$ is determined by $$  \eta - c  = -\int_\eta^\infty \1_I.$$
Our proof of Theorem \ref{ward1} is complete.
\end{proof}

\begin{rmk*} 
Theorem \ref{ward1} shows that for an ``abstract'' translation invariant kernel, Ward's equation only determines it up to a real constant $c$. It is worth pointing out that each limiting $1$-point function $R$ in Theorem \ref{exker} satisfies the condition \eqref{Rleft}. More precisely, one can prove that $R(x)\to 1$ as $x\to -\infty$ from the argument presented in \cite[Section 5]{AKM}.

In the free boundary case when the confinement parameter $\param=1$, we know that $I=(-\infty, 0)$. Indeed, this is shown in \cite{AKM} depending on the interior estimate of $R$ and the boundary fluctuation theorem from \cite{AM}. 
Moreover, we have seen above that
$I=(-\infty, 0)$ for the Ginibre ensemble (with general $\param$) and we believe that this should be a universal fact. In the
following sections, we will prove the statement for concrete scaling limits in (for example) radially symmetric cases, by applying the method of quasipolynomials.
\end{rmk*}

\section{Boundary universality: some preliminaries}
In this section, we collect some preliminary observations which will come in handy in several occasions, when we prove universality of scaling limits
and maximum moduli (Theorem \ref{thm:ml} and Theorem \ref{thm:sr}).
We will also give some background about Laplacian growth.

\subsection{Outline of strategy}

 Let $p_{j,n}$ be the $j$:th orthonormal polynomial with respect to the weight $e^{-nQ^{(\param)}}$, and write $w_{j,n}=p_{j,n}e^{-nQ^{(\param)}/2}$. We start with the basic identity for the 1-point function
$$\bfR_n(\zeta)=\bfK_n(\zeta,\zeta)=\sum_{j=0}^{n-1}|w_{j,n}(\zeta)|^2.$$
As a preliminary step, we shall prove that if $\zeta$ is very close to the outer boundary
$\Gamma=\d \Pc S$, then all terms but the last $\sqrt{n}\log n$ ones can be neglected. More precisely, if $\zeta$ belongs to a \textit{belt}
\begin{equation}\label{belt}N_\Gamma=\{\zeta;\, \dist(\zeta,\Gamma)\le C/\sqrt{n}\},\end{equation}
then
\begin{equation}\label{redu}\bfR_n(\zeta)\sim \sum_{j=n-\sqrt{n}\log n}^{n-1}|w_{j,n}(\zeta)|^2.\end{equation}

The idea is to prove an asymptotic formula for the remaining $w_{j,n}$'s, so that universality will follow after a simple-minded summation. (See Figure \ref{fig3} and Figure \ref{fg:wop}.)

\begin{figure}[ht]
\includegraphics[width=.32\textwidth]{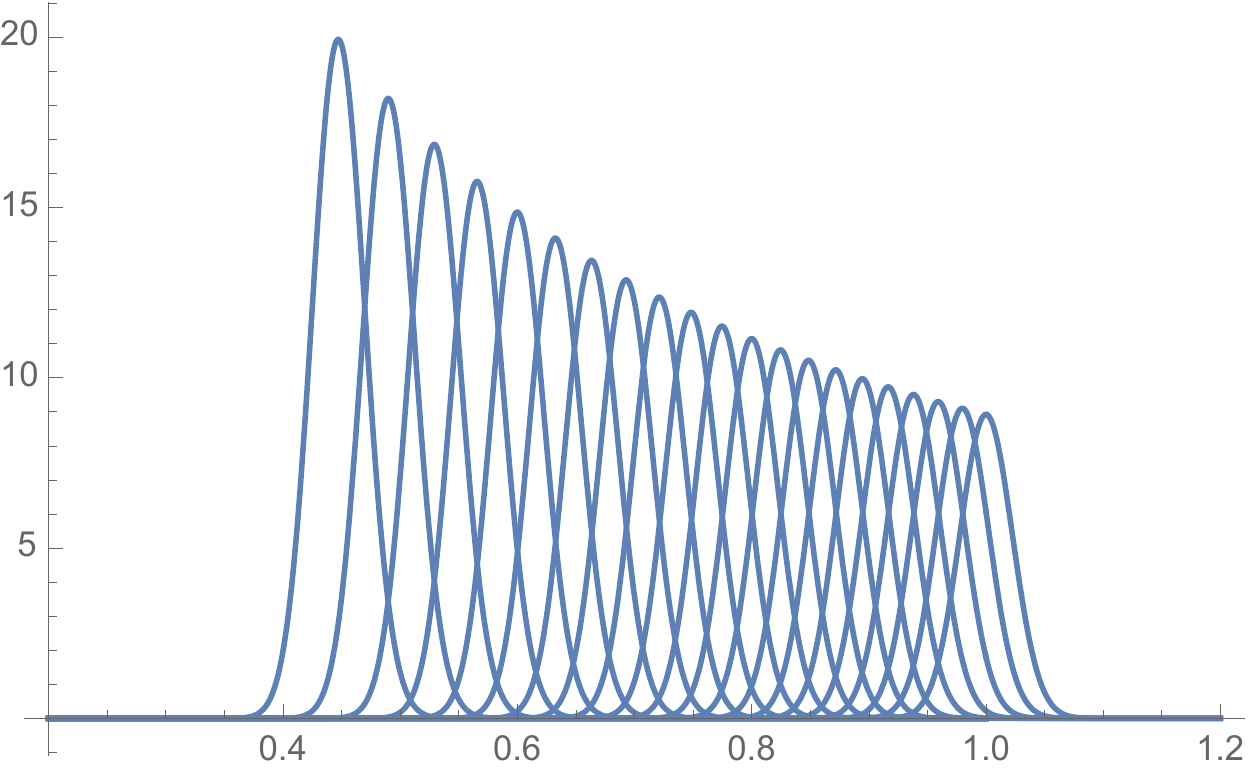}
\includegraphics[width=.32\textwidth]{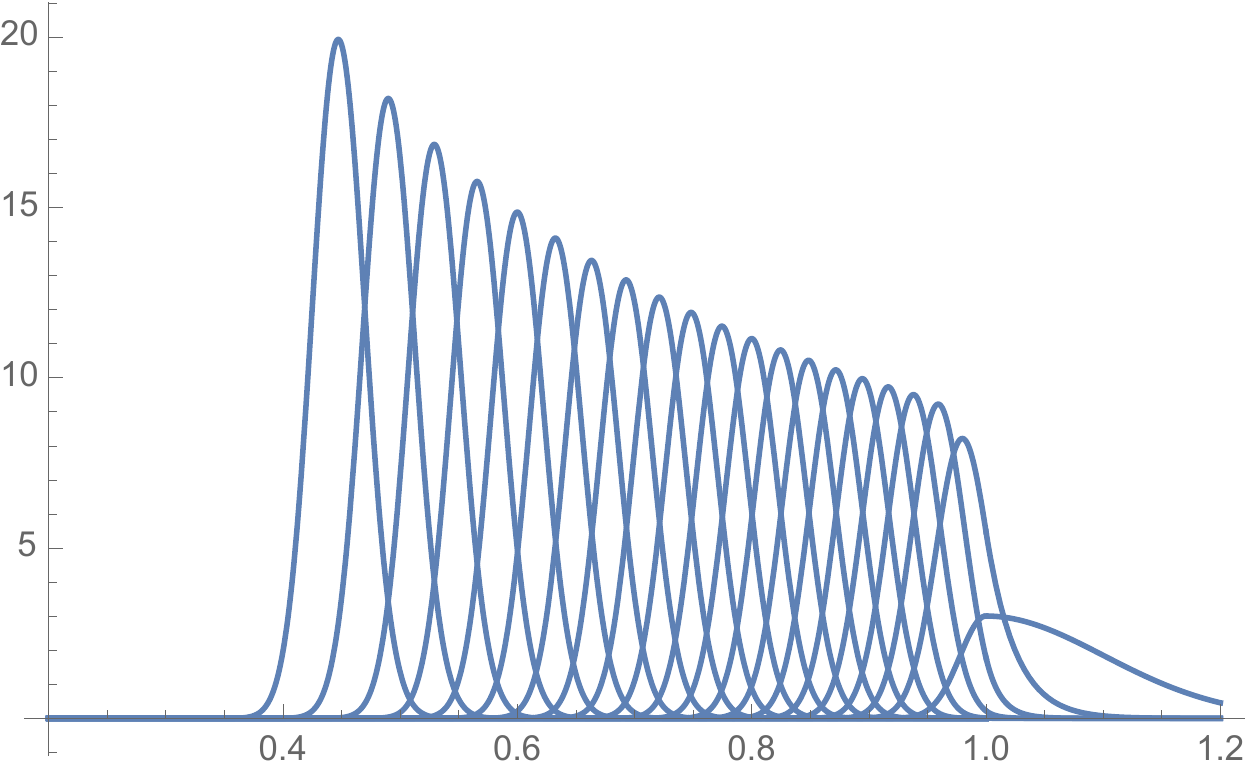}
\includegraphics[width=.32\textwidth]{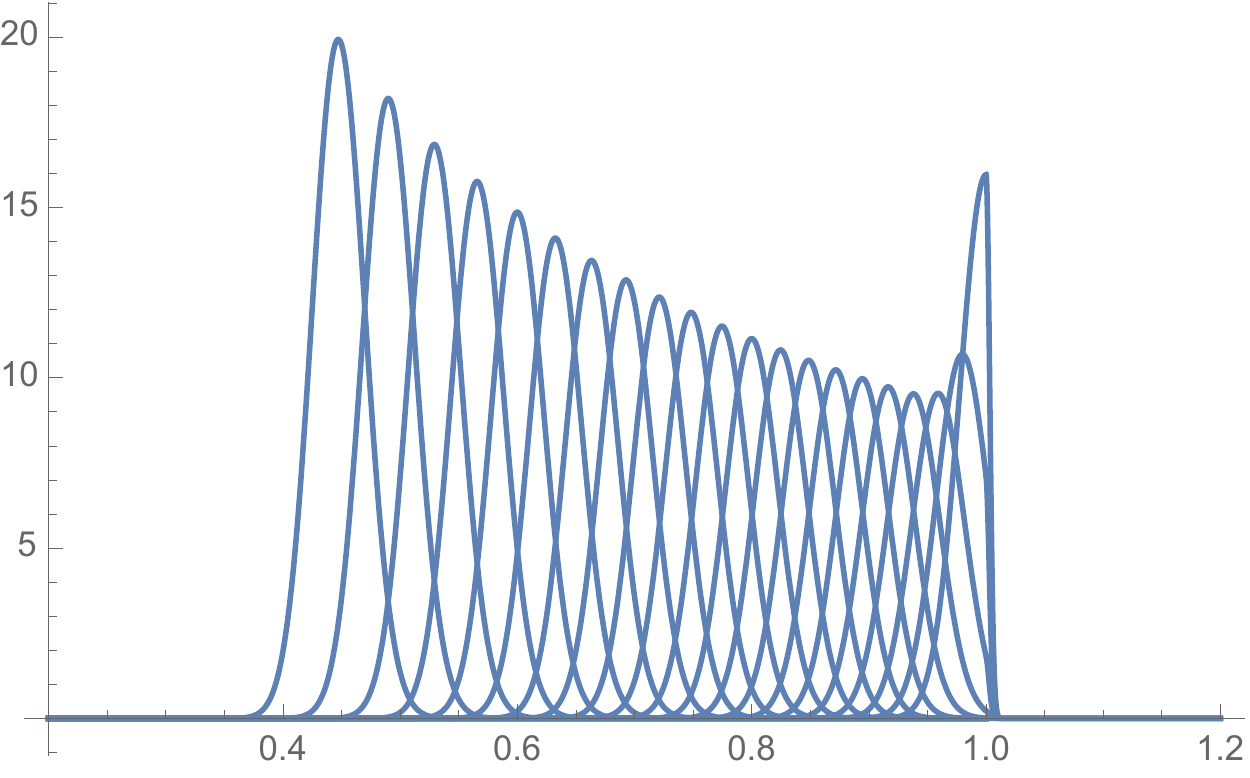}
\caption{Graphs of squared weighted orthonormal polynomials $|w_{j,n}|^2$
when $n=500$ and $j$ varies from $100$ to $500$ in steps of $20$ with different $\param$: $\param = 1$ (left), $\param = 0.05$ (middle), and $\param = 50$ (right). The abscissa $x=1$ is a boundary point of the droplet.}
\label{fg:wop}
\end{figure}

\subsection{Preliminaries on Laplacian growth}\label{prel}

In this subsection we recast some results on Laplacian growth; this is convenient, if not otherwise, to introduce some notation. References and further reading can be found, for instance in \cite{GM,HM,LM,VE} and in \cite[Section 2]{HW17}.

For a parameter $\tau$ with $0<\tau\le 1$, we let
$\eqpot_\tau$ be the \textit{obstacle function} defined by the obstacle $Q$, subject to the growth condition
$$\eqpot_\tau(\zeta)=2\tau\log|\zeta|+O(1),\quad \text{as}\quad \zeta\to\infty.$$
The precise definition runs as follows: for each $\eta\in\C$, $\eqpot_\tau(\eta)$ is the supremum of
$s(\eta)$ where $s$ is a subharmonic function which is everywhere $\le  Q$ and satisfies $s(\zeta)\le 2\tau\log|\zeta|+O(1)$ as $\zeta\to\infty$.

Write $S_\tau$ for the droplet in external potential $Q/\tau$ and note that
$\sigma(S_\tau)=\tau$
where $\sigma$ is the equilibrium measure \eqref{EQMEAS}.
Under our conditions,
$S_\tau$ equals to closure of the interior of the coincidence set $\{Q=\eqpot_\tau\}$ and the measure $\sigma_\tau:=\1_{S_\tau}\cdot \sigma$ minimizes the weighted energy \eqref{q-en} amongst positive measures of total mass $\tau$.

Clearly the droplets $S_\tau$ increase with $\tau$. The evolution of the $S_\tau$'s is known as Laplacian growth.
We will write $\Gamma_\tau$ for the outer boundary,
$$\Gamma_\tau=\d\Pc S_\tau.$$
Hence $\Gamma_\tau=\d U_\tau$ where
$$U_\tau=\hat{\C}\setminus \Pc S_\tau.$$
Finally, we denote by
$$\phi_\tau:U_\tau\to\D_e$$
the unique conformal (surjective) map, normalized so that $\phi_\tau(\infty)=\infty$ and $\phi_\tau'(\infty)>0$. ($\D_e$ is the exterior disk $\{|\zeta|>1\}\cup\{\infty\}$.)

It is well known that $\eqpot_\tau$ is $C^{1,1}$-smooth on $\C$ and harmonic on $U_\tau$. Moreover, since $\Gamma=\Gamma_1$ is everywhere regular, it follows from standard facts about Laplacian growth that $\Gamma_\tau$ is everywhere regular for all $\tau$ in some interval $\tau_0\le \tau\le 1$, where $\tau_0<1$. Below we fix, once and for all, such a $\tau_0$.

\begin{lem} (``Richardson's lemma''). If $0<\tau<\tau'\le 1$ and if $h$ is harmonic in $U_\tau$ and smooth up to the boundary, then
$$\int_{S_{\tau'}\setminus S_\tau}h\Lap Q\, dA=(\tau'-\tau)h(\infty).$$
\end{lem}

\begin{proof} Since $\sigma(S_\tau)=\tau$,
the asserted identity is true when $h$ is a constant. Hence we can assume that $h(\infty)=0$. Let us fix such an $h$ and extend it to $\C$ in a smooth way.
It follows from the properties of the obstacle function $\eqpot_\tau$ that
$$\int_{S_\tau}h\Lap Q\, dA=\int_\C h\Lap \eqpot_\tau\, dA=\int_\C \Lap h\cdot \1_{S_\tau}Q\, dA.$$
Subtracting the corresponding identity with $\tau$ replaced by $\tau'$ we find
$\int_{S_{\tau'}\setminus S_\tau}h\Lap Q\, dA=0,$
finishing the proof of the lemma.
 \end{proof}

The lemma says that if $d\sigma_\tau=\Lap Q\cdot\1_{S_\tau}\, dA$ is equilibrium measure of mass $\tau$, then near the outer boundary component $\Gamma_\tau$, we have, in a suitable ``weak'' sense that
$\frac {d\sigma_\tau} {d\tau}=\omega_\tau$
where $\omega_\tau$ is the harmonic measure of $U_\tau$ evaluated at $\infty$. This means: if $h$ is a continuous function on $\d S_\tau$, then $\omega_\tau(h)=\tilde{h}(\infty)$, where $\tilde{h}$ is the harmonic extension to $U_\tau$ of $h$.

Let us define the Green's function of $U_\tau$ with pole at infinity by
$G(\zeta,\infty)=\log|\phi_\tau(\zeta)|^2$.
Then by Green's identity,
$$\tilde{h}(\infty)=\int_{U_\tau}\tilde{h}(\zeta)\,\Lap G (\zeta,\infty)\, dA(\zeta)=-\frac 1 2 \int_{\d S_\tau} h\frac {\d G}{\d n}\, ds_\tau,$$
where $ds_\tau$ stands for the arclength measure on $\Gamma_\tau$ divided by $2\pi$.

Now for values of $\tau$ such that $\Gamma_\tau$ is everywhere regular
$-\frac \d {\d n} G(\zeta,\infty)=2|\phi'_\tau(\zeta)|$ when $\zeta\in \Gamma_\tau$
so we conclude that
$\frac d {d\tau}(\1_{S_\tau}\Lap Q\, dA)=|\phi'_\tau|ds_\tau$,
meaning that the outer boundary $\Gamma_\tau$ moves in the direction normal to $\Gamma_\tau$, at local speed $|\phi'_\tau|/2\Lap Q$. (The factor $2$ comes about because of
the different normalizations of $dA$ and $ds_\tau$.)

The above dynamic of $\Gamma_\tau$ was of course deduced in a ``weak'' sense, but using the regularity of the curves involved, one can turn it into a pointwise
estimate. More precisely, one has the following result, which is essentially \cite[Lemma 2.3.1]{HW17}.

In the following, we denote by $\normal_\tau$ the outwards unit normal on the curve $\Gamma_\tau$.

\begin{lem} \label{movin} Let $\zeta$ be a point of $\Gamma=\Gamma_1$ and fix
$\tau$ with $\tau_0\le\tau\le 1$. Let $\zeta_\tau$ be the point in $\Gamma_\tau\cap(\zeta+\normal_1(\zeta)\R)$ which is closest
to $\zeta$. Then we have
$$\zeta_{\tau}=\zeta + (\tau-1)\normal_{\tau}(\zeta_\tau)\frac{|\phi_{1}'(\zeta)|}{2\Lap Q(\zeta)}+O((\tau-1)^2),\quad \tau\to 1$$
and $$\normal_{\tau}(\zeta_{\tau}) = \normal_{1}(\zeta)+O(\tau-1), \quad \tau\to 1.$$

In particular there are constants $c_1,c_2>0$ such that for all $\zeta\in \Gamma$ and all $\tau_0\le\tau\le 1$,
$$c_1(1-\tau)\le \dist(\Gamma_\tau,\zeta)\le c_2(1-\tau).$$
\end{lem}

For $\tau_0<\tau\le 1$, we will denote by $V_\tau$ the harmonic continuation of the harmonic function $\eqpot_\tau$ on $U_\tau$ across the analytic curve $\Gamma_\tau$. Considering the growth as $\zeta\to\infty$, we obtain the basic identity
\begin{equation}\label{good}V_\tau(\zeta)=\re\calQ_\tau+\tau\log|\phi_\tau(\zeta)|^2,
\end{equation}
where $\calQ_\tau$ is the holomorphic function on $U_\tau$ with $\re \calQ_\tau=Q$ on $\Gamma_\tau$ and $\im \calQ_\tau(\infty)=0$.

The function $Q-V_\tau$, considered in a neighborhood of the curve $\Gamma_\tau$, plays a central role for the theory.

\begin{lem}\label{lem:Qe}
Fix a number $\tau \in [\tau_0,1]$ and a point $p_\tau \in \Gamma_\tau$. Then we have for $v\in\C$
$$(Q-V_\tau)(p_\tau + v) = 2\Lap Q(p_\tau) \cdot v_{\normal}^2 + O(|v|^3),\quad |v|\to 0,$$
where $v_{\normal}=v\cdot\normal=\re(v\bar{\normal})$ is the component of $v$ in the normal direction $\normal=\rmn_{\tau}(p_\tau)$.
\end{lem}

\begin{proof}
Let $\pa_\rmn$ and $\pa_\rmt$ denote the exterior normal and tangential derivatives on $\Gamma_\tau$, respectively.
Since $Q=V_\tau$ and $\nabla Q = \nabla V_\tau$ on $\Gamma_\tau$, we have for all $\zeta\in\Gamma_\tau$
$$\pa_\rmn(Q-V_\tau)= \pa_\rmt(Q-V_\tau) = \pa_\rmt^2(Q-V_\tau) =0.$$
This implies $\pa_\rmn^2(Q-V_\tau)=4\Lap(Q-V_\tau) = 4\Lap Q$ on $\Gamma_\tau$, finishing the proof of the lemma.
\end{proof}

In particular, it follows from Lemma \ref{lem:Qe} that
there is a number $c>0$ such that, if $\zeta\in N_\Gamma$,
\begin{equation}\label{skurk}(Q-\eqpot_\tau)(\zeta)=(Q-V_\tau)(\zeta)\ge c\dist(\zeta,S_\tau)^2.\end{equation}

\subsection{Discarding lower order terms}\label{disc}
Let $N_\Gamma$ be the belt \eqref{belt} and recall that $Q=Q^{(\param)}$ in $N_\Gamma$. Below we fix an arbitrary point
$\zeta\in N_\Gamma.$
Also fix a number $\tau_0<1$ such that the curves $\Gamma_\tau$ are regular for all $\tau$ with $\tau_0\le \tau\le 1$. Now write
$$\tau=\tau(j)=j/n,\qquad (j\le n-1).$$
It follows from Lemma \ref{l6} that there is a number $C=C_\param$ such that
$$\tau\le \tau_0\quad  \Rightarrow\quad |w_{j,n}(\zeta)|^2\le
Cne^{-n(Q^{(\param)}-\eqpot_{\tau_0})(\zeta)}\le Cne^{-cn},$$
where $c=\inf_{N_\Gamma}\{Q^{(\param)}-\eqpot_{\tau_0}\}>0$.
Hence $$\bfR_n(\zeta)\sim\sum_{j=\tau_0 n}^{n-1}|w_{j,n}(\zeta)|^2.$$

Next fix $j$ such that $\tau_0 n\le j\le n-\sqrt{n}\log n$, i.e.,
$\tau_0\le \tau\le 1-\delta_n$ where we put \textit{throughout}
$$\delta_n=\frac {\log n}{\sqrt{n}}.$$

By Lemma \ref{movin}, we have $\dist(\Gamma_\tau,\Gamma)\ge c(1-\tau)\ge c\delta_n$ for some constant $c>0$.
By \eqref{skurk},
$n(Q^{(\param)}-\eqpot_\tau)(\zeta)\ge c_1\log^2 n$ where $c_1>0$, so
$$|w_{j,n}(\zeta)|^2\le Cne^{-n(Q^{(\param)}-\eqpot_\tau)(\zeta)}\le Cne^{-c_1\log^2 n},\qquad (j\le n-\sqrt{n}\log n).$$
We have shown that
$$\bfR_n(\zeta)\sim \sum_{j=n-\sqrt{n}\log n}^{n-1}|w_{j,n}(\zeta)|^2, \quad (n\to\infty)$$
in the sense that the difference of the left and right sides converges uniformly to zero on $N_\Gamma$, as $n\to\infty$.

\section{Approximation by quasipolynomials} \label{apqq}

In this section, we discuss an asymptotic formula for weighted quasipolynomials $w_{j,n}$ when $j$ is in the range
$n-\sqrt{n}\log n\le j\le n-1$. This formula will be used repeatedly, in our proofs of Theorem \ref{thm:ml} and Theorem \ref{thm:sr}.

The asymptotic formula leads to functions behaving ``essentially'' like
weighted orthogonal polynomials, provided that
certain conditions are satisfied
(e.g., the potential is radially symmetric). However, in the following discussion, we make no assumptions concerning
radial symmetry, except where this is explicitly stated (namely, it is used in the proof of property (P2) below).

Throughout the section, we assume that $n$ is large enough, and we write $\tau=\tau(j)=j/n$ and $\delta_n=n^{-1/2}\log n$.
Thus $1-\delta_n\le \tau< 1.$

\subsection{Basic definitions}
When $n$ is large, $\Gamma_\tau$ is a real-analytic Jordan curve and thus
the map $\phi_\tau:U_\tau\to \D_e$ can be continued analytically across $\Gamma_\tau$ to $\C\setminus K$ where $K$ is a suitable compact set chosen so that $\phi_\tau$ maps $K$ biholomorphically onto $\D_e(\rho_0-\delta)$ for some $\rho_0<1$ and $\delta>0$. ($\D_e(\rho)$ denotes the set $\{|\zeta|>\rho\}\cup\{\infty\}$.)

Let $\calQ_\tau$ be the bounded holomorphic function defined on a neighborhood of $U_\tau$ with $\re \calQ_\tau = Q$ on $\Gamma_\tau$ and $\im \calQ_\tau(\infty) = 0$. (See \eqref{good}.)

For each $j$ with $n-n\delta_n\leq j\leq n-1$, we write $$\xi_{j,n}(\zeta) = \frac{j-n}{\sqrt{n}}\cdot \frac{|\phi_\tau'(\zeta)|}{\sqrt{\Lap Q(\zeta)}},\qquad (\zeta\in\hat{\C}\setminus K).$$
It is convenient to introduce the notation
$$\Phi_{j,n}(\zeta):= \Phi_{\param}(\xi_{j,n}(\zeta)),$$ where $\Phi_{\param}$ is the function defined in \eqref{defphi}.

Finally, we define the \textit{approximate quasipolynomial} $F_{j,n}=F_{j,n,\param}$ of degree $j$ associated with the external field $nQ^{(\param)}$ by
\begin{equation}\label{qop}
F_{j,n} = \left(\frac{n}{2\pi}\right)^{\frac{1}{4}}\sqrt{\phi_\tau'}\,\phi_\tau^{j}\,e^{n\calQ_\tau/2}\,e^{\calH_{j,n}/2},
\end{equation}
where $\calH_{j,n}$ is the bounded holomorphic function defined on $\hat{\C}\setminus K$ such that
$$\re \calH_{j,n} = \log \sqrt{\Lap Q} - \log \Phi_{j,n}\qquad\text{on}\quad \d S_\tau$$
and $\im \calH_{j,n}(\infty)=0$.

We shall show that $F_{j,n}$ satisfies the following properties in Sections \ref{sec:apno} and \ref{sec:ao}. Fix a smooth function $\chi_0$ on $\hat{\C}$ such that $\chi_0=0$ on $K$ and
$\chi_0=1$ on $O:=\phi_\tau^{-1}(\D_e(\rho_0))$.
\begin{itemize}
\item[(P1)] $F_{j,n}$ is approximately normalized:
$$\int_{\C} \chi_{0}^2 \,|F_{j,n}|^2 \, e^{-nQ^{(\param)}}\, dA = 1+O\left(n^{-1/2}\log^2 n\right).$$

\item[(P2)] $F_{j,n}$ satisfies the following ``approximate orthogonality'': for any polynomial $p$ with $\deg p < j$,
$$\int_{\C} \chi_{0}\, F_{j,n}\, \bar{p}\, e^{-nQ^{(\param)} }\,dA = O\left(n^{-1/2}\log^2 n\,\|p\|_{nQ^{(\param)} }\right).$$
\end{itemize}

Following \cite{HW17} we now define a ``positioning operator'' $\Lambda_{j,n}$ by
$$\Lambda_{j,n}[f] = \phi_{\tau}'\cdot\phi_{\tau}^j\cdot e^{n\calQ_\tau/2}\cdot f\circ \phi_\tau.$$
Then the map $\Lambda_{j,n}$ is an isometric isomorphism from $L^2_{nR_{{\param},\tau}}(\D_e(\rho_0))$ to $L^2_{nQ^{(\param)} }(O)$ where $$R_{{\param},\tau} := (Q^{(\param)}  - V_\tau)\circ \phi_\tau^{-1},$$
and it furthermore preserves holomorphicity. (See \cite[Section 3.1]{HW17}.)

The isometry property means that for all $f,\,g\in L^2_{nR_{{\param},\tau}}(\D_e(\rho_0))$
$$\int_{O}\Lambda_{j,n}[f]\,\overline{\Lambda_{j,n}[g]}\, e^{-nQ^{(\param)} }\,dA = \int_{\D_e(\rho_0)}f\,\bar{g}\,e^{-nR_{{\param},\tau}}\,dA.$$
We now define a function $f_{j,n}$ on $\D_e(\rho_0)$ by the formula
$$n^{\frac{1}{4}}\Lambda_{j,n}[f_{j,n}]=F_{j,n}.$$
This gives $$f_{j,n} = (2\pi)^{-1/4}((\phi_{\tau}')^{-1/2}\cdot e^{\calH_{j,n}/2})\circ \phi_\tau^{-1}.$$

\subsection{Approximate normalization}\label{sec:apno}
In this subsection we prove the property (P1) of the approximate quasipolynomials. Recall that $U_\tau=\hat{\C}\setminus \Pc S_\tau$.

We first note that, for points $\zeta$ close to the unit circle $\T$,
\begin{align*}
R_{{\param},\tau}(\zeta)=
\begin{cases}
(Q-V_\tau -(1-{\param})(Q-V_1))\circ\phi_\tau^{-1}(\zeta)
& \text{if}\qquad  \zeta \in \phi_\tau(U_1),  \\
(Q-V_\tau)\circ \phi_\tau^{-1}(\zeta)
& \text{otherwise}.
\end{cases}
\end{align*}

For each $\tau \in [1-\delta_n,1]$, the curve $\phi_\tau(\d S)$ is an analytic Jordan curve, a slight perturbation of the unit circle $\T$.
It is convenient to represent this curve according to the polar parameterization
$$\phi_\tau(\pa S) = \{r_{\tau}(\eta)\cdot\eta : \eta\in \T\},$$
where $r_\tau$ is a smooth function $\T\to\R_+$. (It is easy to see that such a parameterization is possible when $\tau$ is sufficiently close to $1$.)

For a fixed $\eta\in\T$ we now consider the local speed $\dot{r}_\tau(\eta)=\frac {d r_\tau(\eta)}{d\tau}$. Since the conformal map $\phi_\tau$ locally magnifies distances by a
factor $|\phi_\tau'|$ and since the curve $\Gamma_\tau$ moves towards $\Gamma_1$ with local speed $|\phi'_\tau|/(2\Lap Q)$ (see discussion before Lemma \ref{movin}), we find that
$$\dot{r}_\tau(\eta)=\left(\frac {|\phi_\tau'|^2}{2\Lap Q}\right)\circ\phi_\tau^{-1}(\eta)+O(1-\tau).$$
We have shown that
\begin{equation}\label{rtau}
r_\tau(\eta)=1+(1-\tau)\cdot\left(\frac{|\phi_\tau'|^2}{2\Lap Q}\right)\circ\phi_\tau^{-1}(\eta)+O((1-\tau)^2),\quad \tau\to 1.\end{equation}

In the following we will write $$R_\tau:=R_{1,\tau}.$$
It is useful to note that (since $V_\tau$ is harmonic)
$$\Lap R_\tau=\left(\frac {\Lap Q}{|\phi_\tau'|^2}\right)\circ\phi_\tau^{-1}.$$

We shall now analyze the behavior of the function $R_{{\param},\tau}$ close to $\T$.

\begin{lem}\label{lem:ep}
Fix a point $\eta \in \T$. If $r<r_{\tau}(\eta)$,
$$R_{{\param},\tau}(r\eta) = 2\Lap R_\tau(\eta)(r-1)^2 + O((r-1)^3),\qquad r\to 1.$$
If $r>r_{\tau}(\eta)$,
$$R_{{\param},\tau}(r\eta)=2\Lap R_{\tau}(\eta)(r-1)^2 -2(1-{\param})\Lap R_{\tau}(\eta)(r-r_\tau(\eta))^2 + O((r-1)^3),\qquad r\to 1.
$$
\end{lem}

\begin{proof}
If $r<r_\tau (\eta)$, then $R_{{\param},\tau}=(Q-V_\tau)\circ \phi_\tau^{-1}=R_\tau$. From the argument in Lemma \ref{lem:Qe}, we obtain the following Taylor series expansion of $R_{{\param},\tau}$ about $\eta$:
$$R_{{\param},\tau}(r\eta) =2\Lap R_\tau(\eta)(r-1)^2 + O((r-1)^3),\qquad r\to 1.$$
Now assume that $r>r_\tau(\eta)$.
By Lemma \ref{movin}, we have
\begin{equation}\label{nn}
\rmn_1(r_\tau(\eta)\eta) = \eta + O(\tau-1),\quad \tau \to 1,
\end{equation}
where $\rmn_1(\zeta)$ denotes the outer unit normal to $\phi_\tau(\pa S_1)$ at $\zeta\in \phi_\tau(\pa S_1)$.
Hence the second-order Taylor expansion of $(Q-V_1)\circ \phi_\tau^{-1}$ about the point $r_\tau (\eta)\eta \in \phi_\tau(\d S_1)$ takes the form
\begin{equation}\label{tyl}
((Q-V_1)\circ \phi^{-1}_\tau)(r_\tau(\eta)\eta+\xi) = 2\Lap R_\tau (r_\tau(\eta)\eta) (\xi \cdot \rmn_1(r_\tau(\eta)\eta))^2 + O(|\xi|^3),\quad |\xi|\to 0.\end{equation}

Now note that $\Lap ((Q-V_1)\circ \phi_\tau^{-1})= |(\phi_\tau^{-1})'|^2(\Lap Q)\circ\phi_\tau^{-1} = \Lap R_\tau$ since $V_1$ is harmonic. Also note that if $r \to 1$ with $r>r_\tau(\eta)$, then $r_\tau(\eta)- 1 = O(1-r)$ and $1-\tau = O(1-r)$ by \eqref{rtau}.
Now combining \eqref{nn} and \eqref{tyl}, we obtain
\begin{align*}
((Q-V_1)\circ \phi^{-1}_\tau)(r\eta) &= ((Q-V_1)\circ \phi^{-1}_\tau)(r_\tau(\eta)\eta + (r-r_\tau(\eta))\eta)\\
&=2\Lap R_\tau (\eta)\cdot (r_\tau(\eta) - r)^2 + O((r-1)^3),\quad r\to 1.
\end{align*}
Hence the proof is complete.
\end{proof}

Let $D_{j,n}=\{\zeta\in\C: 1-M\delta_n \leq |\zeta| \leq 1+M\delta_n\}$ for a large number $M$. We have the following lemma which will be used to show that our quasipolynomials
are approximately normalized.
\begin{lem}\label{lem:un}
For any integer $j$ with $\tau = j/n \in [1-\delta_n,1)$, we have
$$\int_{D_{j,n}} n^{\frac{1}{2}}\,|f_{j,n}(\zeta)|^2\, e^{-nR_{{\param},\tau}(\zeta)}\,dA(\zeta)=1+O(n^{-1/2}\log ^2n).$$
\end{lem}

\begin{proof}
For each $\eta\in \T$, we compute the integral
$$\int_{1-M\delta_n}^{1+M\delta_n}|f_{j,n}(r\eta)|^2 e^{-nR_{{\param},\tau}(r\eta)} rdr.$$
Fix $\eta\in\T$ and consider the function $R_{{\param},\tau,\eta}(r):=R_{{\param},\tau}(r\eta)$.
For $r<r_{\tau}(\eta)$, using the expansion of $R_{{\param},\tau}$ in Lemma \ref{lem:ep}, we have
\begin{align*}
\int_{1-M\delta_n}^{r_{\tau}(\eta)}|f_{j,n}(r\eta)|^2 e^{-n R_{{\param},\tau,\eta}(r)} rdr
&= \int_{1-M\delta_n}^{r_\tau(\eta)} r|f_{j,n}(r\eta)|^2 e^{-2n \Lap R_\tau(\eta)(r-1)^2 + O(n\delta_n^3)} dr.
\end{align*}
Note that for all $r$ with $|r-1|<M\delta_n$ and all $j$ with $j/n\in [1-\delta_n,1)$
$$|f_{j,n}(r\eta)|^2 = |f_{j,n}(\eta)|^2 + O\left(\delta_n \,{\log n}\right).$$
Recall from the expansion \eqref{rtau} that
$$r_\tau(\eta)-1 = (1-\tau)\cdot (2\Lap R_\tau(\eta))^{-1} + O((1-\tau)^2).$$
Thus the above integral is approximated by
\begin{align}\label{nes1}
\int_{1-M\delta_n}^{r_{\tau}(\eta)}|f_{j,n}(r\eta)|^2 e^{-n R_{{\param},\tau,\eta}(r)} rdr
=\frac{|f_{j,n}(\eta)|^2}{\sqrt{4n\Lap R_{\tau}(\eta)}}\int_{-\infty}^{-\xi_{\tau}(\eta)} \!\!\!\!\! e^{-\frac{1}{2}x^2} dx + E_1,
\end{align}
where $E_1$ is an error term with $E_1 = O(n^{-1}\log^2 n)$ and
$$\xi_\tau (\eta)= \frac{\sqrt{n}(\tau-1)}{\sqrt{\Lap R_\tau(\eta)}}=\frac{j-n}{\sqrt{n}}\left(\frac{|\phi_\tau'|} {\sqrt{\Lap Q}}\right)\circ \phi_\tau^{-1}(\eta).$$
On the other hand, using Lemma \ref{lem:ep} again,  we obtain for $r>r_\tau(\eta)$
$$R_{{\param},\tau,\eta}(r)= 2{\param}\Lap R_\tau(\eta)\left(r-r^*_{\param}\right)^2 +
\frac{2({\param}-1)}{\param}\Lap R_\tau(\eta)(r_\tau(\eta)-1)^2
+  O((r-1)^3),$$
where $r^*_{\param}=(1-(1-{\param})r_\tau(\eta))/{\param}$. It follows from the expansion \eqref{rtau} that
$$R_{{\param},\tau,\eta}(r) = 2{\param}\Lap R_\tau(\eta)\left(r-r^*_{\param}\right)^2 + \frac{{\param}-1}{2{\param}}\cdot\frac{(1-\tau)^2}{\Lap R_\tau(\eta)} +  O((r-1)^3)$$
and $r^*_{\param}$ has the asymptotic expansion of the form
$$r^*_{\param} = 1- \frac{1-{\param}}{\param}\cdot \frac{1-\tau}{2\Lap R_\tau(\eta)}+O((1-\tau)^2).$$
Estimating the integral by Laplace's method, we thus obtain
\begin{align}\label{nes2}
\int_{r_\tau(\eta)}^{1+M\delta_n} \!\! |f_{j,n}(r\eta)|^2
&e^{-n R_{{\param},\tau,\eta}(r)} rdr
= \frac{|f_{j,n}(\eta)|^2}{\sqrt{4n{\param}\Lap R_\tau(\eta)}} e^{\frac{1-{\param}}{2{\param}}\xi_\tau(\eta)^2}\int_{-{\xi_\tau(\eta)}/\sqrt{\param}}^{\infty} e^{-\frac{x^2}{2}}dx + E_2
\end{align}
where $E_2$ is a new error term with $E_2 = O(n^{-1}\log^2 n)$. Combining \eqref{nes1} and \eqref{nes2}, we obtain
\begin{align*}
n^{\frac{1}{2}}\int_{D_{j,n}} |f_{j,n}|^2\, e^{-nR_{{\param},\tau}}\,dA = \sqrt{2\pi}\int_{\T} \frac{|f_{j,n}(\eta)|^2}{\sqrt{\Lap R_\tau(\eta)}} \Phi_{j,n}(\phi_\tau^{-1}(\eta))\,ds(\eta)+E_3 = 1+ E_3
\end{align*}
where $E_3 =O(n^{-1/2}\log^2 n)$ and $ds(\eta) = |d\eta|/2\pi$ is the normalized arclength measure.\end{proof}

We now conclude our proof of the approximate normalization property (P1) of our quasipolynomials.

\begin{lem} \label{ap1}
For any integer $j$ with $\tau = j/n\in [1-\delta_n,1),$ we have
\begin{equation}
\int_{\C} \chi_0^2|F_{j,n}|^2 e^{-nQ^{(\param)}} dA = 1+O(n^{-1/2}\log^2 n).
\end{equation}
\end{lem}
\begin{proof}
By Lemma \ref{lem:un} we obtain
\begin{align*}
\int_{\C} \chi_0^2|F_{j,n}|^2 e^{-nQ^{(\param)}} dA = n^{\frac{1}{2}}\int_{\C\setminus D_{j,n}} (\chi_0 \circ \phi^{-1}_{\tau})^2|f_{j,n}|^2 e^{-nR_{{\param},\tau}} dA + 1+O(n^{-1/2}\log ^2 n).
\end{align*}
In order to estimate the first term on the right, we observe that the functions $f_{j,n}$ are uniformly bounded in $\D_e(\rho_0)$. By Lemma \ref{lem:ep}, there is an $\epsilon>0$ such that for all $\eta$ with $M\delta_n \le \dist(\eta,\T) \le \epsilon $,
$$R_{{\param},\tau}(\eta) \geq c \,\delta_n^2$$
for some constant $c$. Also, by the growth rates of $Q$ and $V_\tau$ near infinity, i.e., $$\liminf_{\zeta\to \infty}\, \frac {Q(\zeta)}{\log |\zeta|^2}>1,\quad \liminf_{\zeta\to \infty}\frac {V_\tau(\zeta)} {\log |\zeta|^2}=\tau,$$
there exists a constant $C>0$ such that
$$|\eta|\ge 1+\epsilon\qquad \Rightarrow\qquad R_{{\param},\tau}(\eta) \geq C \log |\eta|.$$
Summing up, we obtain
\begin{equation}\label{extint}
n^{\frac{1}{2}}\int_{\C\setminus D_{j,n}} (\chi_0 \circ \phi^{-1}_{\tau})^2|f_{j,n}|^2 e^{-nR_{{\param},\tau}} dA \leq O(n^{\frac{1}{2}}e^{-c'\log^2 n})\end{equation} for some constant $c'$, which completes the proof.
\end{proof}

\subsection{Approximate orthogonality}\label{sec:ao} Up to this point we have not assumed that $Q$ be radially symmetric.
However, in order to prove the following lemma, we shall use this assumption.

\begin{lem}\label{ap2}
Suppose that $Q$ is radially symmetric. Fix an integer $j$ with $\tau = j/n \in [1-\delta_n,1)$.  Let $p$ be a holomorphic polynomial of degree $\ell$ less than $j$.
Then $$\int_{\C} \chi_{0}\, p\, \overline{F_{j,n}} e^{-nQ^{(\param)}} dA = O(n^{1/4}e^{-c(\log n)^2}\|p\|_{nQ^{(\param)}}).$$
\end{lem}

\begin{proof}
Write $q = \Lambda_{j,n}^{-1}[p]$. Then $q$ is holomorphic on $\D_e(\rho_0)$ and satisfies $q(\eta)= O(|\eta|^{\ell-j})$ as $\eta\to \infty$.
We first show that
\begin{equation}\label{ordn}
\int_{D_{j,n}}q \,\overline{f_{j,n}} \, e^{-nR_{{\param},\tau}}dA = 0.\end{equation}
Define a function $h$ by $h = q/f_{j,n}$. Then $h$ is holomorphic in $\D_e(\rho_0)$ and vanishes at infinity since $f_{j,n}$ does not vanish at infinity. Now we write the above integral as
$$\int_{1-M\delta_n}^{1+M\delta_n}\int_{\T} 2h \,|f_{j,n}|^2\,e^{-nR_{{\param},\tau}} r ds dr.$$
Since $|f_{j,n}|^2\,e^{-nR_{{\param},\tau}}$ is radially symmetric and $\int_{\T} h \,ds = h(\infty) = 0$ by mean-value theorem, the above integral vanishes and we obtain \eqref{ordn}.

For the remaining part, using Cauchy-Schwarz inequality, we have
\begin{align*}
\left|\int_{\C\setminus\phi^{-1}_\tau(D_{j,n})} \chi_0 \, p \, \overline{F_{j,n}} e^{-nQ^{(\param)}} dA \right| \leq \|p\|_{nQ^{(\param)}} \left(\int_{\C\setminus\phi^{-1}_{\tau}(D_{j,n})}|\chi_0|^2 |F_{j,n}|^2 e^{-nQ^{(\param)}} dA\right)^{1/2},
\end{align*}
which gives the bound $O(n^{1/4}e^{-c(\log n)^2}\|p\|_{nQ^{(\param)}})$ for some $c$ by \eqref{extint}.
\end{proof}

\begin{rmk*} Lemma \ref{ap2} says that the quasipolynomials $F_{j,n}$ satisfy the approximate orthogonality property (P2)
with a better error bound. In the sequel we will not use radial symmetry of $Q$, but merely that the property (P2) holds.
\end{rmk*}

\subsection{Pointwise estimates}
Going back to the general case, we assume that $F_{j,n}$ satisfies the properties (P1) and (P2). (But we do not necessarily assume that $Q$ be radially symmetric.)

\begin{lem} Let $p_{j,n}$ the $j$:th orthonormal polynomial in $L^2(\C, e^{-nQ^{(\param)}}dA)$. Then for all $j$ with $j/n\in [1-\delta_n,1)$
$$\|p_{j,n} - F_{j,n}\chi_0\|_{nQ^{(\param)}} = O(n^{-1/2}\log^2 n).$$
\end{lem}
\begin{proof}
Let $u_0$ be the $L^2(e^{-nQ^{(\param)}})$ norm-minimal solution to the $\dbar$-problem
$$\dbar u = F_{j,n}\dbar \chi_0$$ which satisfies $u_0(\zeta)=O(|\zeta|^{j-1})$ as $\zeta\to\infty$.
By a standard H\"{o}rmander estimate (cf. \cite[Section 4.2]{Hor}), we have
\begin{equation}\label{ine1}
\|u_0\|_{nQ^{(\param)}}^2 \leq Cn^{-1}\int_{\C} |\dbar\chi_0 F_{j,n}|^2   e^{-nQ^{(\param)}} = O(e^{-cn}).
\end{equation}
Here, the $O(e^{-cn})$ bound is obtained from the fact that $\bp \chi_0$ vanishes on a neighborhood of $\Gamma_\tau$ and $|F_{j,n}|^2 e^{-nQ^{(\param)}}$ has an exponential decay (in $n$) on the support of $\dbar \chi_0$.

Define a function $\tilde{p}_{j,n}$ by
$$\tilde{p}_{j,n} = F_{j,n}\, \chi_0 - u_0.$$
Then $\tilde{p}_{j,n}$ is entire and has the polynomial growth $\tilde{p}_{j,n}(\zeta) = a |\zeta|^{j}+O(|\zeta|^{j-1})$ with some nonzero $a$ near infinity. It follows that $\tilde{p}_{j,n}$ is a polynomial of exact degree $j$. By \eqref{ine1}
$$\|\tilde{p}_{j,n}-F_{j,n}\,\chi_0\|_{nQ^{(\param)}} \leq O(e^{-cn}).$$
We thus obtain from (P1) and (P2) that
\begin{equation}\label{epp1}
\|\tilde{p}_{j,n}\|_{nQ^{(\param)}} = 1+O(n^{-1/2}\log^2 n)\end{equation}
and
\begin{equation}\label{epp2}
\int_{\C}\tilde{p}_{j,n} \,\bar{p}\, e^{-nQ^{(\param)}} dA = O(n^{-1/2} \log^2 n \,\|p\|_{nQ^{(\param)}})\end{equation}
for all polynomials $p$ of degree $\leq j-1$.

Now we consider the orthogonal projection $\pi_{j,n}$ from $L^2(e^{-nQ^{(\param)}})$ onto the space consisting of all  polynomials in $L^2(e^{-nQ^{(\param)}})$ of degree $\leq j-1$.
Then the function
$$p^{\star}_{j,n}:=\tilde{p}_{j,n} - \pi_{j,n}[\tilde{p}_{j,n}]$$ is a polynomial degree $j$ which satisfies by \eqref{epp2}
$$\|p^{\star}_{j,n}-\tilde{p}_{j,n}\|_{nQ^{(\param)}}=O(n^{-1/2}\log^2 n).$$ Since $p^{\star}_{j,n}$ is orthogonal to all polynomials of degree at most $j-1$, it can be written as
$p^{\star}_{j,n}=c_{j,n}\,p_{j,n}$ for some constant $c_{j,n}$. Since $\|p^{\star}_{j,n}\|_{nQ^{(\param)}} = 1+O(n^{-1/2}\log^2 n)$ by \eqref{epp1} and \eqref{epp2}, we have
$|c_{j,n}|= 1+O(n^{-1/2}\log^2 n)$. We can assume that $c_{j,n}$ is a positive real number. Hence we get $$\|p_{j,n}-F_{j,n}\chi_0\|_{nQ^{(\param)}} = O(n^{-1/2}\log^2 n).$$
\end{proof}

\begin{lem} \label{lem:pwq}
Suppose that $$\dist(\zeta, \Gamma_\tau )\leq n^{-1/2}(\log \log n)^{1/2}.$$
Then we have the approximation for all $j$ with $j/n\in [1-\delta_n,1)$
$$p_{j,n}(\zeta) = F_{j,n}(\zeta)(1+O(n^{-\beta})),$$
where $\beta>0$ and the $O$-constant is uniform in $\zeta$.
\end{lem}

\begin{proof}
For $\zeta$ with $\dist(\zeta,\Gamma_\tau) \leq n^{-1/2} (\log\log n)^{1/2}$, we have
$$(\check{Q}_\tau - V_\tau)(\zeta) \leq c \dist(\zeta, \Gamma_\tau)^2 \leq c\, n^{-1} \log\log n.$$
The estimate in Lemma \ref{l6} gives that there exists a constant $C>0$ such that
\begin{align*}
|p_{j,n} - F_{j,n}|\leq C\sqrt{n}\, \|p_{j,n}-F_{j,n}\chi_0\|_{nQ^{(\param)}}\,e^{n\check{Q}_\tau/2},
\end{align*}
which implies there exists a constant $c'>0$ such that
$$|p_{j,n}(\z)-F_{j,n}(\z)| = O((\log n)^2 e^{n\check{Q}_\tau(\z)/2}) = O((\log n)^{c'} e^{n V_\tau(\z)/2}).$$ Since $$|F_{j,n}| = n^{1/4}|\sqrt{\phi_\tau'}|\,e^{n V_\tau/2} e^{\re\calH_{j,n}/2}$$ and $\re \calH_{j,n}$ is bounded,
the lemma is proved.
\end{proof}

\section{Error function approximation}\label{sec:err}
In this section, we prove Theorem \ref{thm:ml}. We assume that conditions (P1) and (P2) of Section \ref{apqq} are satisfied. (E.g.,
$Q$ is radially symmetric.)

Take a point $p$ on the outer boundary $\Gamma$ of $S$ and consider points $\zeta$ in the disc $|\zeta-p|\le M/\sqrt{n}$ for some
large constant $M$. As we observed in Section \ref{prel}, the 1-point function rescaled about $p$ can be written as
\begin{equation}\label{r2sum}
R_n(z) =R_n^\sharp(z)+o(1),\qquad R_n^\sharp(z)= \frac{e^{-nQ^{(\param)}(\zeta)}}{n\Lap Q(p)}\sum_{j=m_n}^{n-1}|p_{j,n}(\zeta)|^2
\end{equation}
where $m_n = n-n\delta_n$ and $\zeta = p+\rmn_1(p)\,z/\sqrt{n\Lap Q(p)}$.  To obtain the asymptotics of $R_n^\sharp(z)$, we
shall apply the quasipolynomial approximation in Section \ref{apqq}.

Consider a point $\zeta$ with $\dist(\zeta,U_\tau)\le n^{-1/2}(\log\log n)^{1/2}$. By Lemma \ref{lem:pwq}, we obtain that for all $j$ with $n-n\delta_n \leq j \leq n-1$
\begin{align*} |p_{j,n}(\zeta)|^2 e^{-nQ^{(\param)}(\zeta)}
&= |F_{j,n}(\zeta)|^2 e^{-nQ^{(\param)}(\zeta)}(1+O(n^{-\beta})),\\
&= \left(\frac{n}{2\pi}\right)^{1/2}|\phi_\tau'| e^{-n(Q^{(\param)} - V_\tau)} e^{\re \calH_{j,n}}(1+O(n^{-\beta})).
\end{align*}
Here the error term is uniform for all $j$ with $n- n\delta_n \leq j \leq n-1$. Recall that $\calH_{j,n}$ is a bounded holomorphic function on the relevant set of $\zeta$ satisfying
 $$\re\calH_{j,n}(\zeta) = \log \sqrt{\Lap Q(\zeta)} - \log \Phi_{\param}(\xi_{j,n}(\zeta)),\quad \zeta\in \pa S_\tau$$ where $\Phi_{\param}$ is the function defined in \eqref{defphi} and 
$\xi_{j,n}(\zeta) = (j-n)|\phi_\tau'(\zeta)|/\sqrt{n\Lap Q(\zeta)}$.
\begin{lem}\label{lem:2r}
We have that
$$\lim_{n\to \infty} R_n^\sharp(z) = S(2\re z) \, e^{2(1-{\param})(\re z)_+^2},\quad z\in \C.$$
Here the convergence is bounded on $\C$ and locally uniform in $\C\setminus i\R$.
\end{lem}
\begin{proof}
We first observe that
\begin{equation*}
R_n^\sharp(z) = \frac{1}{\sqrt{2\pi}}\sum_{k=1}^{n\delta_n} \frac{|\phi_\tau'(\zeta)|}{\sqrt{n}\Lap Q(p)} e^{-n(Q^{(\param)} - V_\tau)(\zeta)} e^{\re \calH_{j,n}(\zeta)} (1+ O(n^{-\beta})),
\end{equation*}
where $k = n -j$ and $\zeta = p+ \rmn_1(p)\, z/\sqrt{n\Lap Q(p)}$.
For any compact subset $\calD\subset \C$, we have
$$|\phi_\tau'(\zeta)|\ e^{\re\calH_{j,n}(\zeta)} = \frac{|\phi_\tau'(p)|\ \sqrt{\Lap Q(p)}}{\Phi_{\param}(\xi_{j,n}(p))}+ O(n^{-1/2}\log n),\quad n\to\infty,$$
where the error term is uniform for $z\in \calD$ and for $j$ with $n-n\delta_n \leq j \leq n-1$. To obtain the asymptotic expansion of $Q-V_\tau$, let $p_\tau$ denote the closest point to $p$ where the line $p+ \rmn_1(p)\,\R$ intersects the boundary $\pa S_\tau$. Then in view of Lemma \ref{movin} the following asymptotic formulas hold: as $\tau \to 1$,
$$p_\tau = p + (\tau -1) \,\rmn_1(p) \ \frac{|\phi_\tau'(p)|}{2\Lap Q(p)} + O((1-\tau)^2), \quad \rmn_\tau(p_\tau) = \rmn_1(p)+O(1-\tau).$$ The Taylor series expansion about $p_\tau$ in Lemma \ref{lem:Qe} gives that for $\tau \in [1-\delta_n,1)$
\begin{align*}
(Q-V_\tau)(\zeta)& = (Q-V_\tau) \Big(p_\tau +  \rmn_\tau(p_\tau) \frac{k}{n} \frac{|\phi_\tau'(p)|}{2\Lap Q(p)}  + \frac{\rmn_\tau(p_\tau)\,z}{\sqrt{n\Lap Q(p)}} +O(n^{-1}\log^2 n) \Big) \\
&= 2\Lap Q(p_\tau) \left(\frac{k}{n}\frac{|\phi_\tau'(p)|}{2\Lap Q(p)} + \frac{\re z}{\sqrt{n\Lap Q(p)}}\right)^2 + O(n^{-3/2}\log^3 n),
\end{align*}
where the error is uniform for $z \in \calD$ and for $\tau \in [1-\delta_n,1)$. We also note that
$$n(Q-V_1)(\zeta) = 2(\re z)^2 + O(n^{-1/2}),\quad z\in \calD.$$

Take a compact subset $\calD_1$ of $\LL$. Then for sufficiently large $n$ the points $$\zeta = p+ \rmn_1(p)\, z/\sqrt{n\Lap Q(p)},\ z\in \calD_1 $$ are contained in the droplet $S$. Thus for $z\in \calD_1$, we have
$$e^{-n(Q^{(\param)} - V_\tau)(\zeta)} = e^{-n(Q-V_\tau)(\zeta)} = e^{-\left(-\xi_{j,n}(p) + 2\re z\right)^2/2}(1+o(1)).$$
On the other hand, for a compact subset $\calD_2$ of $\C\setminus \overline{\LL}$, it holds that for $z \in \calD_2$
\begin{align*}
e^{-n(Q^{(\param)} - V_\tau)(\zeta)}&= e^{-n(Q-V_\tau)(\zeta) + n(1-{\param})(Q-V_1)(\zeta)}\\
& = e^{-(-\xi_{j,n}(p)+2\re z)^2/2} e^{2(1-{\param})(\re z)^2}(1+o(1)).
\end{align*}

Combining the above asymptotics, we obtain
\begin{equation*}
R_n^\sharp(z)= \frac{1}{\sqrt{2\pi}}\frac{|\phi_\tau'(p)|}{\sqrt{n\Lap Q(p)}}\sum_{k=1}^{\sqrt{n}\log n} \frac{ e^{-\frac{1}{2}(\xi_{j,n}(p)-2\re z)^2} e^{2(1-{\param})(\re z)_+^2 }} {\Phi_{\param}(\xi_{j,n}(p))}
(1+o(1)).
\end{equation*}
Recall that $\xi_{j,n}(p) = -{k}|\phi_\tau'(p)|/{\sqrt{n\Lap Q(p)}}$. The sum can be considered as an approximate Riemann sum with step length $|\phi_\tau'(p)|/\sqrt{n\Lap Q(p)}$, which implies
\begin{equation*}
R_n^\sharp(z) = \int_{-\infty}^{0} \frac{e^{-\frac{1}{2}(\xi - 2\re z)^2} }{\Phi_{\param}(\xi)} \frac{d\xi}{\sqrt{2\pi}}\cdot e^{2(1-{\param})(\re z)_+^2 }\cdot (1+o(1)),\qquad (n\to\infty).
\end{equation*}
Hence the lemma is proved.
\end{proof}

\begin{proof}[Proof of Theorem \ref{thm:ml}]
Recall that by Theorem \ref{exker} every limit $K$ of correlation kernels $c_n K_n$ is of the form
$$K(z,w) = G(z,w) \Psi(z,w) e^{(1-{\param})((\re z)_+^2 +(\re w)_+^2)},$$
where $\Psi$ is a Hermitian-entire function. As a result of Lemma \ref{lem:2r} we obtain the convergence $$\lim_{n\to \infty} R_n(z) = S(2 \re z)\, e^{2(1-{\param})(\re z)_+^2 }, \quad z\in\C.$$
This implies $\Psi(z,z) = S(2\re z).$
By analytic continuation, we obtain $\Psi(z,w) = S(z+\bar{w})$, completing our proof of Theorem \ref{thm:ml}.
\end{proof}

\section{Scaling limit of the maximal modulus}\label{sec:max}
In this section we discuss the asymptotic distribution of the maximal modulus at the edge and prove Theorem \ref{thm:sr}. For this purpose, we
pick a random sample $\{\zeta_j\}_1^n$ (associated with a modified potential $Q^{(\param)}$) and recall that the maximal modulus is defined by
$$|\zeta|_n = \max_{1\leq j \leq n} |\zeta_j|.$$

Observe that $|\zeta|_n \leq r$ if and only if
none of the points $\zeta_j$ belongs to the exterior disk $\D_e(r)$. Thus we must investigate the gap probability
that no point belongs to $\D_e(r)$.

It is well known that the distribution function of $|\zeta|_n$ is represented by
\begin{align*}
\Prob_n(|\zeta|_n \leq r) = \det\left(\delta_{j,k} - \int_{\D_e(0,r)} p_{j,n}\,\overline{p}_{k,n}e^{-nQ^{(\param)}} dA\right)_{j,k=0}^{n-1},
\end{align*}
where $p_{j,n}$ is an orthonormal polynomial of degree $j$ in $L^2 (\C,e^{-nQ^{(\param)}} dA)$. Cf. \cite[Section 3]{R03} or \cite[Section 15.1]{Meh}.

Here and in what follows we use the shorthand $\Prob_n$ to denote the Boltzmann-Gibbs measure $\Prob_{n,\param}^{1}$.

In the case that $Q$ is radially symmetric, the above
probability reduces to
\begin{equation}\label{pbpd}
\Prob_n(|\zeta|_n \leq r) = \prod_{j=0}^{n-1}\left(1- \int_{\D_e (0,r)} |p_{j,n}|^2 e^{-nQ^{(\param)}} dA\right).
\end{equation}

Note that the outer boundary of the droplet is a circle centered at the origin. Let $\rho$ be the radius of that circle.
We rescale $|\zeta|_n$ about $\rho$ by
$$\omega_n = \sqrt{4n{\param} \gamma_n \Lap Q(\rho) }\left(|\zeta|_n - \rho - \sqrt{\frac{\gamma_n}{4n{\param}\Lap Q(\rho)}}\right),$$
where $\gamma_n = \log(n/2\pi)-2\log\log n + \log (\rho^2\Lap Q(\rho)/\Phi_{\param}^2(0))$.

Now we compute the distribution function of the random variable $\omega_n$
$$\Prob_n(\omega_n \leq x) = \Prob_n\left(|\zeta|_n \leq \rho + \frac{1}{\sqrt{4n{\param}\Lap Q(\rho)}}\left(\sqrt{\gamma_n}+\frac{x}{\sqrt{\gamma_n}}\right)\right).$$
It is convenient to write, for $x\in\R$, $$ h_n(x) =   \frac{r_n({\param},x)}{\sqrt{4n\Lap Q(\rho)}}; \qquad r_n({\param},x) = \frac{1}{\sqrt{\param}}\left(\sqrt{\gamma_n} + \frac{x}{\sqrt{\gamma_n}}\right).$$
Then by \eqref{pbpd} we have
\begin{align}\label{ppom}
\Prob_n(\omega_n \leq x) = \prod_{j=0}^{n-1}\left(1- \int_{|\zeta|> \rho + h_n(x)} |p_{j,n}|^2 e^{-nQ^{(\param)}} dA\right).
\end{align}

\begin{lem}\label{lem:sr}
We have the convergence
$$\lim_{n\to \infty}\sum_{j=0}^{n-1} \int_{|\zeta|>\rho + h_n(x)} |p_{j,n}|^2 e^{-nQ^{(\param)}} dA = e^{-x}.$$
Here the convergence is uniform in any compact subset of $\R$.
\end{lem}
\begin{proof}
For each $x$, $\rho + h_n(x)$ is located in $\C \setminus S$ with $O(n^{-1/2}\sqrt{\log n})$-distance from the boundary $\pa S$.
By the growth assumption on $Q$ and the estimate in Lemma \ref{lem:Qe} we have for some $c>0$  $$Q^{(\param)}(\zeta) - V_\tau (\zeta) \geq c \min \{\dist(\zeta, \Gamma_\tau)^{\,2},\, \log(1+|\zeta|) \},\quad \zeta \in U_\tau,$$
where we remind of the notation $U_\tau=\hat{\C}\setminus \Pc S_\tau$ and $\Gamma_\tau=\d U_\tau$.

From this and Lemma \ref{l6} we infer that the sum of lower degree terms
$$\sum_{j \leq n-n\delta_n} \int_{|\zeta|>\rho + h_n(x)} |p_{j,n}|^2 e^{-nQ^{(\param)}} dA $$
 is negligible. (Here $\delta_n = n^{-1/2}\log n$, as usual.)

Using the quasipolynomial approximation from Section \ref{apqq}, we now analyze the limit of the sum of higher degree terms. We will write
$$D_{n,x} = \{\zeta\in\C : \rho + h_n(x)\leq |\zeta| \leq \rho + M\delta_n \}.$$
The problem is reduced to computing, for each $x\in\R$, the sum
\begin{equation}\label{hsint}
\sum_{j= m_n}^{n-1}\int_{D_{n,x}} |F_{j,n}|^2 e^{-nQ^{(\param)}} dA = \sqrt{\frac{n}{2\pi}}\sum_{j= m_n}^{n-1} \int_{\rho + h_n(x)}^{\rho + M\delta_n}|\phi_\tau'| e^{-n(Q^{(\param)} - V_\tau)} e^{\re \calH_{j,n}} 2r dr,\end{equation}
where $m_n = n-n\delta_n$.
By the Taylor series expansion for $Q-V_\tau$ and $Q-V_1$ about $\Gamma_\tau$ and $\Gamma$ respectively, we get for all $r$ with $r-\rho=O(\delta_n)$ and $\tau$ with $\tau \in [1-\delta_n,1)$
\begin{align*}
(Q^{(\param)} -V_\tau)(r)&=(Q-V_\tau)(r) - (1-{\param})(Q-V_1)(r)\\
&=2\Lap Q(\rho_\tau)(r-\rho_\tau)^2-2(1-{\param})\Lap Q(\rho) (r-\rho)^2+O(\delta_n^3)
\end{align*}
where $\rho_\tau$ is the radius of the circle $\Gamma_\tau$.
By the change of variable $$s =\sqrt{4n\Lap Q(\rho)}(r -\rho)$$ and the argument in Section \ref{sec:err},
 the above integral \eqref{hsint} is equal to
\begin{equation}\label{imsum}
 \sum_{j=m_n}^{n-1}\frac{\rho|\phi_\tau'(\rho)|}{\sqrt{2\pi}\Phi_{\param}(\xi_{j,n})}\int_{r_n({\param},x)}^{M'\log n} e^{(1-{\param})s^2/2} e^{-(s-\xi_{j,n})^2/2} ds\, (1+o(1))\end{equation}
where $M'$ is a positive constant and $\xi_{j,n}=\xi_{j,n}(\rho) = (j-n)|\phi_\tau'(\rho)|/\sqrt{n\Lap Q(\rho)} $.
By the Riemann sum approximation as in Section \ref{sec:err}, we obtain for $k=n-j$ and $s$ with $r_n(\param,x)\leq s\leq M'\log n$
$$\frac{|\phi_\tau'(\rho)|}{\sqrt{n\Lap Q(\rho)}}\sum_{k=1}^{n\delta_n} \frac{e^{-(s-\xi_{j,n})^2/2}}{\Phi_{\param}(\xi_{j,n})}=\int_{-\log n}^{0}\frac{e^{-(s-\xi)^2/2}}{\Phi_{\param}(\xi)}d\xi + \epsilon_n(s).$$
The error term $\epsilon_n(s)$ from the Riemann sum approximation has a bound
$$\epsilon_n(s) \leq C  \,\frac{(\log n)^2}{\sqrt{n}}\sup_{-\log n \leq \xi \leq 0} \left|\frac{d}{d\xi} \left(\frac{e^{-(s-\xi)^2/2}}{\Phi_{\param}(\xi)}\right)\right|$$ so that $\epsilon_n(s)$ is negligible in \eqref{imsum}. Hence we need to find the limit of the integral
$$\sqrt{\frac{n\Lap Q(\rho)}{2\pi}}\,\rho\int_{r_n({\param},x)}^{M'\log n} e^{(1-{\param})s^2/2} \int_{-\infty}^{0} \frac{e^{-(s-\xi)^2/2}}{\Phi_{\param}(\xi)}d\xi\, ds.$$
Write
$$E_n(x) = \int_{r_n({\param},x)}^{\infty} e^{(1-{\param})s^2/2} \int_{-\infty}^{0} \frac{e^{-(s-\xi)^2/2}}{\Phi_{\param}(\xi)}d\xi\, ds.$$
By the change of variables $u=s-\xi/\param$ and $v=u-r_n(\param,x)$,
we obtain
\begin{align*}
E_n(x) &= \int_{-\infty}^{0} \frac{e^{\frac{1-{\param}}{2{\param}}\xi^2}}{\Phi_{\param}(\xi)} \int_{r_n({\param},x)-\xi/{\param}}^{\infty} e^{-\frac{\param}{2}u^2} du\, d\xi \\
&=\int_{r_n({\param},x)}^{\infty} e^{-\frac{\param}{2}u^2}\int_{{\param}(r_n({\param},x)-u)}^{0} \frac{e^{\frac{1-{\param}}{2{\param}}\xi^2}}{\Phi_{\param}(\xi)}d\xi \,du = \int_{0}^{\infty} e^{-\frac{\param}{2}(v+r_n({\param},x))^2} f(v) dv,
\end{align*}
where
$$f(v) = \int_{-{\param}v}^{0} \frac{e^{\frac{1-{\param}}{2{\param}}\xi^2}}{\Phi_{\param}(\xi)}d\xi.$$
Now the integration by parts gives
\begin{align}\label{intpar}
\int_{0}^{\infty} e^{-\frac{\param}{2}v^2-{\param}v\,r_n({\param},x)}f(v)dv &=  -\frac{1}{{\param}r_n({\param},x)}\left[e^{-{\param}v\,r_n({\param},x)}e^{-{\param}v^2/2}f(v)\right]_{v=0}^\infty \\
& + \int_{0}^{\infty} \frac{e^{-{\param}v\,r_n({\param},x)}}{{\param}\,r_n({\param},x)}e^{-{\param}v^2/2}\left(-{\param}v f(v) + f'(v)\right)dv. \nonumber
\end{align}
Since $f(0)=0$ and $e^{-\param v^2/2}f(v)=o(1)$ as $v\to \infty$, the first term in the right-hand side of \eqref{intpar} vanishes.
By integrating by parts again we obtain
\begin{align*}
\int_{0}^{\infty} e^{-\frac{\param}{2}v^2-{\param}v\,r_n({\param},x)}f(v)dv &= - \frac{1}{({\param}\,r_n({\param},x))^2}\left[e^{-{\param}v\,r_n({\param},x)}e^{-{\param}v^2/2}f_1(v)\right]_{v=0}^\infty \\
&+\int_{0}^{\infty} \frac{e^{-{\param}v\,r_n({\param},x)}}{\left({\param}\,r_n({\param},x)\right)^2}e^{-{\param}v^2/2}\left(-{\param}v f_1(v) + f_1'(v)\right)dv
\end{align*}
where $f_1(v) = f'(v) - {\param}vf(v)$. Since $f_1(0)=f'(0)={\param}(\Phi_{\param}(0))^{-1}$,
\begin{align*}
\int_{0}^{\infty} e^{-\frac{\param}{2}v^2-{\param}v\,r_n({\param},x)}f(v)dv &= \frac{1}{{\param}\Phi_{\param}(0) r_n({\param},x)^2} + O(r_n({\param},x)^{-3}).
\end{align*}
It follows that
\begin{align*}
E_n(x) = e^{-\frac{\param}{2}r_n({\param},x)^2} \left( \frac{1}{\param\Phi_{\param}(0) r_n({\param},x)^2} + O(r_n({\param},x))^{-3} \right).
\end{align*}
Recall that $r_n({\param},x) = \frac{1}{\sqrt{\param}}\left(\sqrt{\gamma_n}+ \frac{x}{\sqrt{\gamma_n}}\right)$ where $$\gamma_n = \log(n/2\pi)-2\log\log n + \log (\rho^2\Lap Q(\rho)/\Phi_{\param}^2(0)).$$
Thus we obtain
$$\sqrt{\frac{n\Lap Q(\rho)}{2\pi}} \rho E_n(x) = e^{-x} + o(1),$$
where $o(1)\to 0$ locally uniformly as $n\to\infty$.
Hence we prove the lemma.
\end{proof}

It follows from \eqref{ppom} that
$$\log \Prob_n(\omega_n \leq x)=-\sum_{j=0}^{n-1} \int_{|\zeta|>\rho + h_n(x)} |p_{j,n}|^2 e^{-nQ^{(\param)}} dA + o(1).$$
By Lemma \ref{lem:sr}, we finally obtain the convergence
$$\Prob_n(\omega_n \leq x) \to e^{-e^{-x}},\quad n\to \infty.$$

Our proof of Theorem \ref{thm:sr} is complete. $\qed$

\section{Summary and outlook}\label{sec:sum}

We have introduced a scale of boundary confinements for the Coulomb gas and analyzed them in the planar, determinantal case.
In particular, we have proved existence of a new scale of point fields, and investigated them for universality in two different ways;
(1) using the method of Ward equations and (2) using a quasipolynomial approximation formula.

The first method works well for
any reasonable potential, but has the drawback that we need to impose some apriori conditions on a solution in order to draw
relevant conclusions.

The method of quasipolynomials becomes more complicated than in the free boundary case, depending on that the quasipolynomial approximation
formula has different forms inside and outside of the droplet. This makes it complicated to verify the approximate orthogonality property (P2)
in Section \ref{apqq}. An exception occurs in cases when the curves $\Gamma_\tau$ has everywhere a uniform distance to $\Gamma$, meaning in practice that we are working with a radially symmetric situation. In this case the method works well, and allows us to conclude universality of scaling limits and to calculate the distribution of the maximum modulus.

We hope that our two methods will contribute to a future resolution of the intriguing question of universality in the above setting.

It is well known that fluctuations about the equilibrium measure of a determinantal Coulomb gas $\{\zeta_j\}_1^n$ converges, e.g., in the sense of distributions, to a Gaussian field with free boundary conditions, see \cite{AM,KM}. It is also known (see \cite{AKMW}) that if one imposes
hard edge conditions, then the fluctuations will converge to another kind of Gaussian field.
The question of convergence of fluctuations with an arbitrary value of the confinement-parameter $\param$ presents itself naturally.

In the paper \cite{A19}, a notion of ``distance to the vacuum'' was introduced and analyzed for $\beta$-ensembles. This gives a natural
generalization of the maximum modulus, which works also for non radial potentials. The paper \cite{A19} provides some estimates
of the distribution function for this distance; cf. the paper \cite{CF} for numerical simulations and more precise predictions.

Finally, we wish to point out that our construction can in principle be applied to a Coulomb gas in any dimension, or, say, in a compact
Riemannian manifold.
Indeed, given a suitable external potential, it is possible to define natural notions of interaction kernel, droplet, obstacle function, and so on
(cf. e.g., \cite{CGZ} and \cite{GR}).
By introducing a confinement-parameter as above, one obtains more general ensembles with various restrictions near the
boundary.

\end{document}